\documentclass[12pt,a4paper,amsmath,amssymb]{amsart}

\allowdisplaybreaks

\usepackage[latin1]{inputenc}
\usepackage{amsthm}
\usepackage{latexsym}
\usepackage{amsfonts}
\usepackage{bbm,bm,dsfont}
\usepackage{eufrak}
\usepackage{color} 
\usepackage{graphicx}
\usepackage{sidecap}
\usepackage{adjustbox}
\usepackage[abs]{overpic}

\newtheorem{proposition}{Proposition}
\newtheorem{theorem}{Theorem}

\newtheorem{corollary}{Corollary}

\theoremstyle{definition}
\newtheorem{definition}{Definition}

\newtheorem{remark}{Remark}

\newcommand{\mc}[1]{\mathcal{#1}}
\newcommand{\ms}[1]{\mathsf{#1}}
\newcommand{\mf}[1]{\mathfrak{#1}}

\newcommand{\N}{\mathbb N}
\newcommand{\R}{\mathbb R}
\newcommand{\Z}{\mathbb Z}
\newcommand{\C}{\mathbb C}

\newcommand{\T}{\mathbb T}

\newcommand{\hil}{\mathcal{H}} 
\newcommand{\tr}[1]{\mathrm{tr}\left[#1\right]} 
\def\<{\langle}
\def\>{\rangle}

\newcommand{\id}{\mathbbm{1}} 
\newcommand{\fii}{\varphi}

\newcommand{\vN}{\otimes_{\rm vN}}

\begin{document}

\title{Joint measurements through quantum broadcasting}

\author{Erkka Haapasalo}
\address{Centre for Quantum Technologies, National University of Singapore, Science Drive 2, Block S15-03-18, Singapore 117543}
\email{cqteth@nus.edu.sg}

\maketitle

\begin{abstract}
We study generating joint measurements by operating on the input quantum state with a broadcasting channel followed by local measurements on the two outputs of the broadcasting channel. Although, due to perfect broadcasting or cloning being impossible, this scheme cannot generate perfect joint observables for all pairs of quantum observables, we study for what intended joint observables and local measurements this procedure can be carried out. For local sharp observables this is always possible as long as the intended joint observable is dominated by the tensor product of the local sharp observables. This result excludes the case of joint measurements of continuous sharp observables with themselves which we indeed show to be impossible to be generated through broadcasting. We also derive necessary and sufficient conditions for the success of the broadcasting scheme for local noisy observables and intended joint observables which are easily checked when the local observables are convex mixtures of sharp observables and white noise. We also study broadcasting of local fuzzy observables in the presence of symmetries. Finally we introduce the scenario of measuring joint observables where we have access to broadcasting, local measurements and post-processing and characterize the most resourceful local measurements in this setting.
\end{abstract}

\section{Introduction}

According to quantum theory, not all pairs of observables can be measured jointly. This {\it incompatibility} of quantum measurements clearly distinguishes quantum physics from classical physics where all physical quantities can be measured simultaneously and, indeed, have objective observer-independent values. Incompatibility or complementarity of quantum observables explains measurement uncertainty relations \cite{BaGrTo2017,BaGrTo2018,BuLaWe2014a,BuLaWe2014b} giving intrinsic lower bounds for the accuracy of joint or sequential measurements. On the other hand, incompatibility can also be seen as a resource as incompatible observables are required for detecting all quantum steering \cite{UoMoGu2014}.

No cloning or, more generally, no broadcasting is another well-known non-classical feature of quantum physics. This means that quantum states cannot be physically mapped into states of a compound system where the reduced states are copies of the original state. Were this possible, all pairs of quantum observables would actually be jointly measurable: just perfectly broadcast the initial state and then locally measure the observables which you wish to jointly measure. Although this method is not universally possible due to no broadcasting, this still provides a way to approximately measure quantum observables jointly. We know that even incompatible pairs of quantum observables can be approximately jointly measured in an optimal way, e.g.,\ with the least amount of noise introduced in the target observables measured by robustness measures \cite{DeFaKa2019,Haapasalo2015}. The question is, how general the broadcasting scheme for generating approximate joint observables for pairs of quantum observables is. Can we, e.g.,\ reach the lower bounds for the noise in joint measurements as characterized by the robustness measures?

We will see that we may generate any (approximate) joint observables for sharp quantum observables using the broadcasting method as long as the target joint observable is dominated by the tensor product of the sharp observables. A notable case where this absolute continuity condition does not hold is the exact joint observable for the joint measurement of a continuous sharp observable with itself. We will indeed see that the broadcasting scheme fails to generate these joint observables. This means that the broadcasting method is not a universal way of generating joint observables. Despite this no-go result (which only arises in the continuous case), broadcasting followed by local measurements remains an intuitive way of joint measuring. We illustrate this by describing broadcasting schemes for generating optimal joint measurements for position and momentum (finite and continuous) and mutually unbiased bases in general.

After introducing the basics of quantum measurements in Section \ref{sec:prel}, we describe how to generate joint measurements through broadcasting and local sharp measurements and derive the above no-go result in Section \ref{sec:jointPVM}. In Section \ref{sec:noisyPOVMs}, we derive necessary and sufficient conditions under which a given joint observable can be generated through broadcasting from given local (fuzzy) observables. These conditions greatly simplify when the local observables are convex combinations of sharp observables and white noise. Until now, we have concentrated only on the question whether the broadcasting scheme succeeds for given local observables and an intended joint observable, but in Section \ref{sec:examples} we study what the broadcasting channels generating optimal joint observables for relevant local measurements actually look like. As examples we consider finite and continuous position and momentum observables. In these cases, the broadcasting channels can be interpreted as parts of standard quantum measurements. As many physically relevant measurement settings reflect natural symmetries, we study broadcasting of local symmetric fuzzy observables into a symmetric joint observable in Section \ref{sec:symm}. We will see that in many situations, the broadcasting channel can also be assumed to be symmetric and that, especially in the case of finite symmetries, we can give simple necessary and sufficient conditions for broadcastability. In Section \ref{sec:BLMPP}, we study which (joint) observables can be reached from fixed local measurements when we have access to different broadcasting schemes and final classical data manipulation (post-processing). We see that rank-1 sharp observables reach the largest classes of observables in this setting and, finally see, in Subsection \ref{subsec:BLMPPMUB} that optimal joint measurements of mutually unbiased bases in particular can be reached from canonical rank-1 sharp observables using the broadcasting, local measurements, and post-processing protocol.

\section{Preliminaries and basic definitions}\label{sec:prel}

We use the convention $\N=\{1,\,2,\,3,\ldots\}$. For any (complex) Hilbert space $\hil$, we denote, respectively, by $\mc L(\hil)$, $\mc T(\hil)$, and $\mc U(\hil)$ the algebra of bounded linear operators on $\hil$, the set of trace-class operators on $\hil$, and the group of unitary operators on $\hil$. We denote the unit (identity operator) of $\mc L(\hil)$ by $\id_\hil$ and use notations $R,\,S,\,T.\ldots\in\mc L(\hil)$ for bounded operators. For trace-class operators, we use notations $\rho,\,\sigma,\,\tau,\ldots\in\mc T(\hil)$. We denote by $\mc S(\hil)$ the subset of positive $\rho\in\mc T(\hil)$ such that $\tr{\rho}=1$ (i.e.,\ $\rho$ is of unit trace).

We present here a quick overview of quantum measurement theory studied in more detail, e.g.,\ in \cite{kirja}. For a Hilbert space $\hil$ and a measurable space $(X,\mc A)$ (i.e.,\ $X$ is a non-empty set and $\mc A$ is a $\sigma$-algebra of subsets of $X$) we say that a map $\ms M:\mc A\to\mc L(\hil)$ is a {\it normalized positive-operator-valued measure (POVM)} if $\ms M(A)\geq0$ for all $A\in\mc A$, $\ms M(\emptyset)=0$, $\ms M(X)=\id_\hil$, and $\ms M\big(\cup_{i=1}^\infty A_i\big)=\sum_{i=1}^\infty\ms M(A_i)$, where the series is defined weakly, whenever $A_1,\,A_2,\ldots\in\mc A$ is a disjoint sequence. POVMs describe the state-dependent statistics of a quantum measurement and thus they are identified with observables. The probability of registering an outcome in a set $A\in\mc A$ in a measurement of a POVM $\ms M$ when the pre-measurement state is $\rho\in\mc S(\hil)$ is $p^{\ms M}_\rho(A):=\tr{\rho\ms M(A)}$. If $\ms M(A)$ is an orthogonal projection for all $A\in\mc A$, we say that the POVM $\ms M$ represents a {\it sharp observable} and we call $\ms M$ as a {\it normalized projection-valued measure (PVM)}.

Measurements of quantum observables (POVMs) cannot typically carried out simultaneously, but when this is possible (by measuring a third observable `containing' the two POVMs), we say that the observables are jointly measurable:

\begin{definition}
Let $\hil$ be a Hilbert space and $(X,\mc A)$ and $(Y,\mc B)$ be measurable spaces and define the measurable space $(X\times Y,\mc A\otimes\mc B)$ where $\mc A\otimes\mc B$ is the $\sigma$-algebra of $X\times Y$ generated by the product sets $A\times B$ where $A\in\mc A$ and $B\in\mc B$. We say that POVMs $\ms M:\mc A\to\mc L(\hil)$ and $\ms N:\mc B\to\mc L(\hil)$ are {\it jointly measurable} if they have a {\it joint observable} $\ms G:\mc A\otimes\mc B\to\mc L(\hil)$ which gives $\ms M$ and $\ms N$ as its margins $\ms G^{(1)}$ and $\ms G^{(2)}$, i.e.,\ for all $A\in\mc A$ and $B\in\mc B$,
$$
\ms M(A)=\ms G^{(1)}(A):=\ms G(A\times Y),\qquad\ms N(B)=\ms G^{(2)}(B):=\ms G(X\times B).
$$
\end{definition}

We will also need the concept of a discrete POVM. If the set $X$ is at most countably infinite, we say that a POVM $\ms M:2^X\to\mc L(\hil)$ is discrete. By defining $M_x:=\ms M(\{x\})$ for all $x\in X$, we can identify the discrete POVM $\ms M$ with a sequence $(M_x)_{x\in X}$ of positive operators such that $\sum_{x\in X}M_x=\id_\hil$ where the series converges weakly. Let $\ms M=(M_x)_{x\in X}$ and $\ms N=(N_y)_{y\in Y}$ be discrete POVMs in a Hilbert space $\hil$. A simple application of the above definition yields that $\ms M$ and $\ms N$ are jointly measurable if and only if there is a discrete POVM $\ms G=(G_{x,y})_{(x,y)\in X\times Y}$ in $\hil$ such that
$$
M_x=\sum_{y'\in Y}G_{x,y'},\qquad N_y=\sum_{x'\in X}G_{x',y}
$$
for all $x\in X$ and $y\in Y$ where the series converge weakly.

Let $\hil$ and $\mc K$ be Hilbert spaces. We say that a linear map $\Phi:\mc T(\hil)\to\mc T(\mc K)$ is an {\it operation} if its Heisenberg dual $\Phi^*:\mc L(\mc K)\to\mc L(\hil)$, defined through $\tr{\rho\Phi^*(R)}=\tr{\Phi(\rho)R}$ for all $\rho\in\mc T(\hil)$ and $R\in\mc L(\mc K)$, is completely positive, i.e.,\ for all $n\in\N$, $R_1,\ldots,\,R_n\in\mc L(\mc K)$, and $\fii_1,\ldots,\,\fii_n\in\hil$,
$$
\sum_{i,j=1}^n\<\fii_i|\Phi^*(R_i^*R_j)\fii_j\>\geq0.
$$
An operation is called a {\it channel} if it is trace preserving.

In this treatise, we are looking at the following procedure to approximately measure two POVMs $\ms M:\mc A\to\mc L(\hil)$ and $\ms N:\mc B\to\mc L(\hil)$ simultaneously: We pick a broadcasting channel $\Phi:\mc T(\hil)\to\mc T(\hil\otimes\hil)$ and broadcast the initial system state $\rho$ with it. Now we measure $\ms M$ on the first subsystem of this broadcasted state and $\ms N$ on the second. The outcome statistics $p:\mc A\otimes\mc B\to[0,1]$ are now given by $p(A\times B)=\tr{\Phi(\rho)\big(\ms M(A)\otimes\ms N(B)\big)}$ for all $A\in\mc A$ and $B\in\mc B$. Using the dual channel $\Phi^*:\mc L(\hil\otimes\hil)\to\mc L(\hil)$ this means that this procedure can be understood as the measurement of the POVM $\ms G:\mc A\otimes\mc B\to\mc L(\hil)$, $\ms G(A\times B)=\Phi^*\big(\ms M(A)\otimes\ms N(B)\big)$ for all $A\in\mc A$ and $B\in\mc B$. Scenario like this has been previously studied in conjunction with sequential measurements and disturbance in \cite{Heinosaari2016}.

\begin{definition}
Let $(X,\mc A)$ and $(Y,\mc B)$ be measurable spaces and $\mc K_1$ and $\mc K_2$ be a Hilbert spaces. For POVMs $\ms M:\mc A\to\mc L(\mc K_1)$ and $\ms N:\mc B\to\mc L(\mc K_2)$, define the POVM $\ms M\otimes\ms N:\mc A\otimes\mc B\to\mc L(\mc K_1\otimes\mc K_2)$ through $(\ms M\otimes\ms N)(A\times B)=\ms M(A)\otimes\ms N(B)$ for all $A\in\mc A$ and $B\in\mc B$. We say that, for another Hilbert space $\hil$, a POVM $\ms G:\mc A\otimes\mc B\to\mc L(\hil)$ {\it can be generated through broadcasting from $\ms M$ and $\ms N$} if there is a quantum channel $\Phi:\mc T(\hil)\to\mc T(\mc K_1\otimes\mc K_2)$ such that $\ms G=\Phi^*\circ(\ms M\otimes\ms N)$, i.e.,
$$
\ms G(A\times B)=\Phi^*\big(\ms M(A)\otimes\ms N(B)\big),\qquad A\in\mc A,\quad B\in\mc B.
$$
\end{definition}

\begin{figure}
\begin{center}
\begin{overpic}[scale=0.25,unit=1mm]{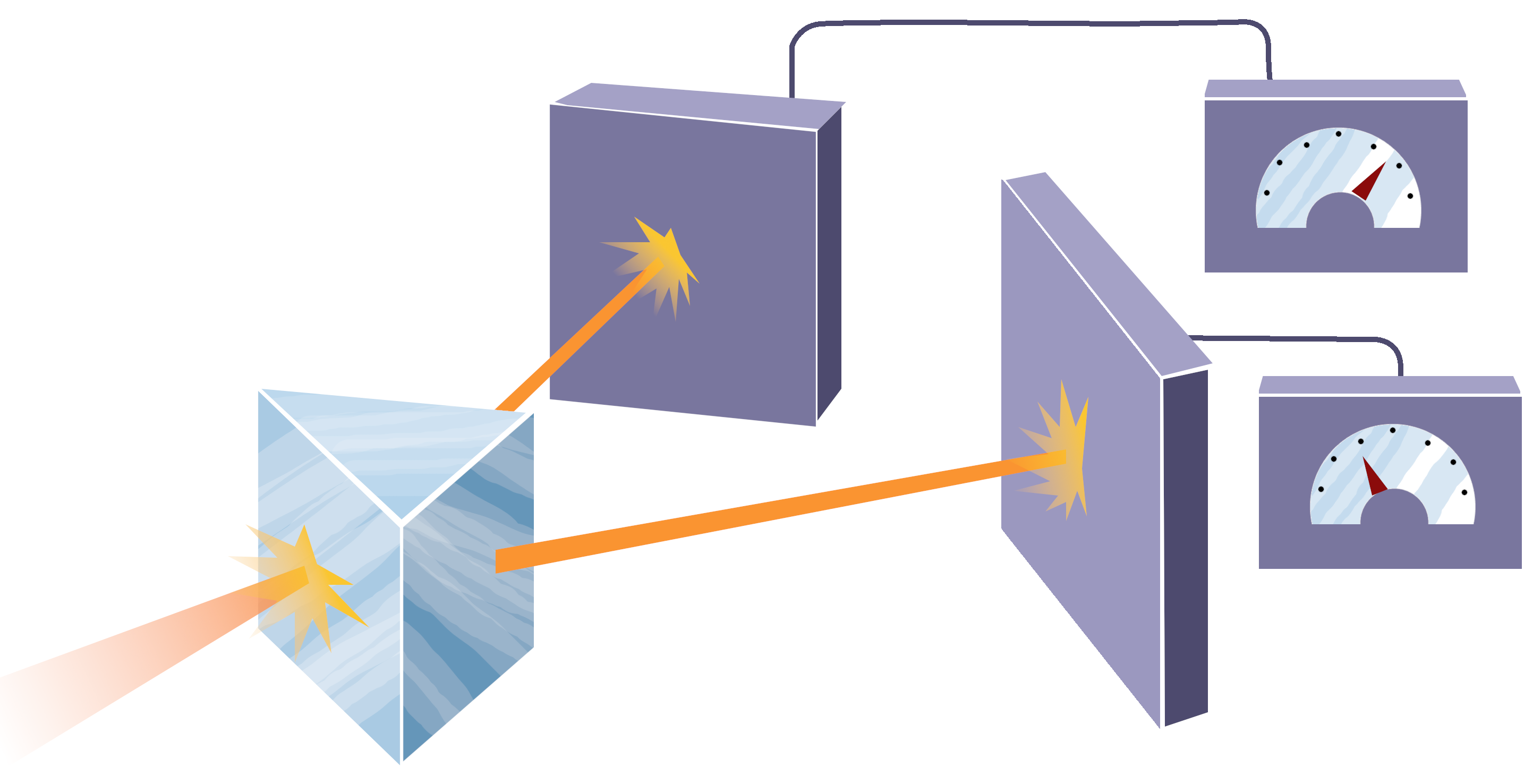}
\put(0,3){\begin{Large}
$\rho$
\end{Large}}
\put(6,13){\begin{Huge}
$\Phi$
\end{Huge}}
\put(18,17){\begin{Large}
$\rho'$
\end{Large}}
\put(31,12){\begin{Large}
$\rho''$
\end{Large}}
\put(53,28){\begin{Large}
$\ms M$
\end{Large}}
\put(53,3){\begin{Large}
$\ms N$
\end{Large}}
\end{overpic}
\includegraphics[scale=0.35]{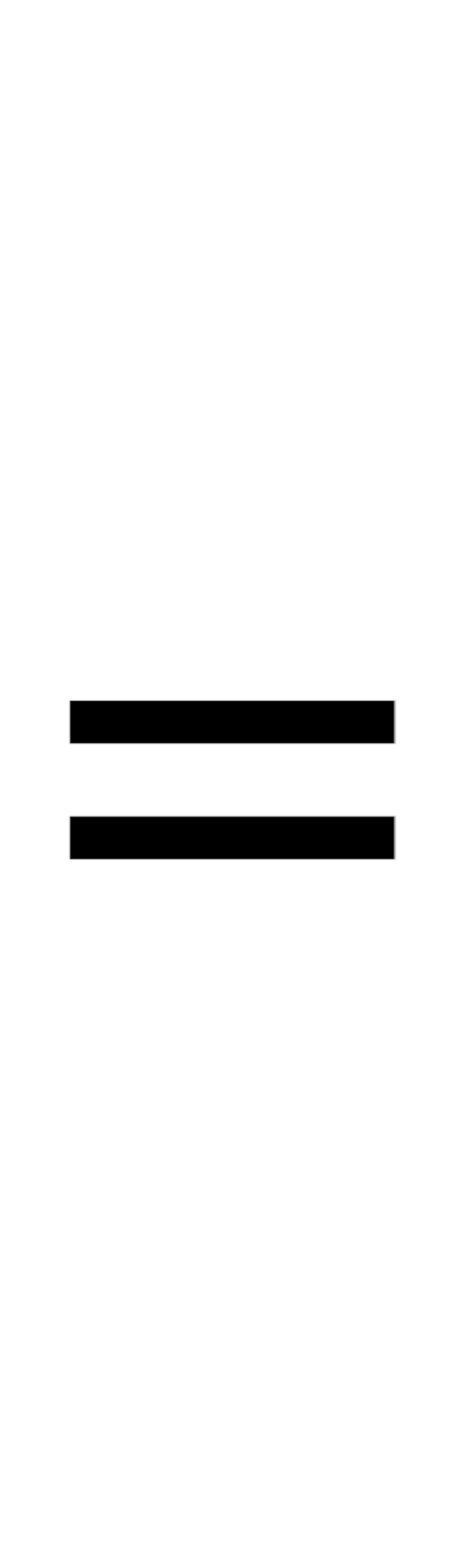}
\begin{overpic}[scale=0.25,unit=1mm]{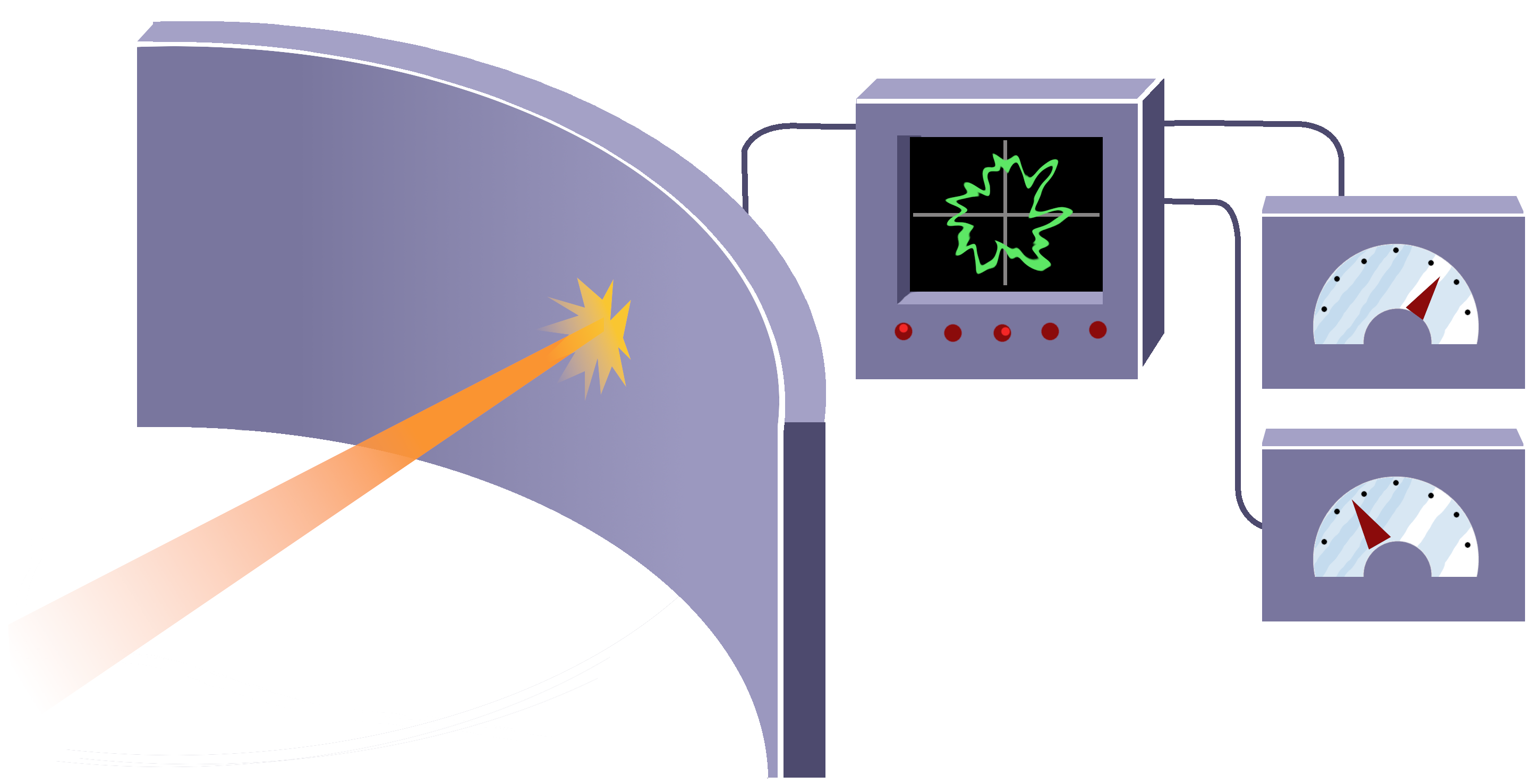}
\put(0,4){\begin{Large}
$\rho$
\end{Large}}
\put(38,10){\begin{Large}
$\ms G$
\end{Large}}
\put(48,26){\begin{large}
$\ms G^{(1)}\sim\ms M$
\end{large}}
\put(48,0){\begin{large}
$\ms G^{(2)}\sim\ms N$
\end{large}}
\end{overpic}
\caption{\label{fig:kloonimittaus} The input state $\rho$ is broadcast by the channel $\Phi$ (here represented by a beam-splitting crystal) and the approximate copies $\rho'$ and $\rho''$ are then measured separately: the POVM $\ms M$ is measured on $\rho'$ and $\ms N$ on $\rho''$. This procedure is equivalent to performing a joint measurement $\ms G$ whose margins $\ms G^{(1)}$ and $\ms G^{(2)}$ are intended as approximations of $\ms M$ and $\ms N$, respectively, or, formally, $\ms G(A\times B)=\Phi^*\big(\ms M(A)\otimes\ms N(B)\big)$ for all value sets $A$ for $\ms M$ and $B$ for $\ms N$. This is what we mean when we say that $\ms G$ is generated from $\ms M$ and $\ms N$ through broadcasting.}
\end{center}
\end{figure}
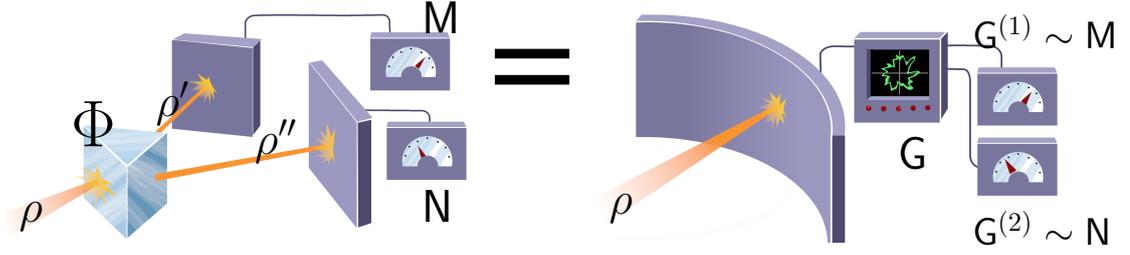

We typically have $\mc K_1=\mc K_2=\hil$ in the above definition. Let $\ms M$ and $\ms N$ be POVMs and $\Phi$ be a broadcasting channel as above. Since perfect quantum cloning or broadcasting is impossible, $\ms M$ typically differs from $\Phi^{(1)\,*}\circ\ms M$ and $\ms N$ differs from $\Phi^{(2)\,*}\circ\ms N$ where $\Phi^{(1)}$ is the composition of $\Phi$ and the partial trace over the second system and $\Phi^{(2)}$ is the composition of $\Phi$ and the partial trace over the first system. However, if $\Phi$ is a reasonably good cloner for the POVMs $\ms M$ and $\ms N$ in some sense, the POVM $\ms G=\Phi^*\circ(\ms M\otimes\ms N)$ might be a reasonably good approximate joint POVM for $\ms M$ and $\ms N$. The question is can any (approximate) joint POVM for $\ms M$ and $\ms N$ be realized through such a broadcasting scheme. We will see shortly that, if $\ms M$ and $\ms N$ are PVMs and $\hil$ is separable, this is true in many cases. However, in many canonical cases, joint measurements through cloning will always fail. These problems arise exclusively in the continuous case.

\section{Joint measurements for PVMs through broadcasting}\label{sec:jointPVM}

The next result basically tells us that any (approximate) joint measurement for two PVMs can be constructed through a broadcasting scheme provided that the joint measurement is suitably well behaved w.r.t.\ the two PVMs. Here and in sequel, we denote $\ms M\ll\ms N$ for any POVMs $\ms M:\mc A\to\mc L(\hil)$ and $\ms N:\mc A\to\mc L(\mc K)$, where $(X,\mc A)$ is some measurable space and $\hil$ and $\mc K$ are Hilbert spaces, if, whenever $A\in\mc A$ is such that $\ms N(A)=0$, then $\ms M(A)=0$ as well. Either one of $\ms M$ and $\ms N$ can also be a scalar measure.

\begin{theorem}\label{theor:PVMbroadcast}
Let $\hil$, $\mc K_1$, and $\mc K_2$ be separable Hilbert spaces, $(X,\mc A)$ and $(Y,\mc B)$ be measurable spaces, and $\ms P:\mc A\to\mc L(\mc K_1)$ and $\ms Q:\mc B\to\mc L(\mc K_2)$ be PVMs. Any POVM $\ms G:\mc A\otimes\mc B\to\mc L(\hil)$ such that $\ms G\ll\ms P\otimes\ms Q$ can be generated from $\ms P$ and $\ms Q$ through broadcasting.
\end{theorem}

\begin{proof}
We start by fixing faithful states $\sigma_i\in\mc S(\mc K_i)$ for $i=1,\,2$ and defining the probability measures $\mu:\mc A\to[0,1]$ and $\nu:\mc B\to[0,1]$ through $\mu(A)=\tr{\sigma_1\ms P(A)}$ for all $A\in\mc A$ and $\nu(B)=\tr{\sigma_2\ms Q(B)}$ for all $B\in\mc B$. It follows that $\ms P$ and $\mu$ are mutually absolutely continuous as are $\ms Q$ and $\nu$. This means that we may extend $\ms P$ into a normal *-isomorphism $\Gamma_{\ms P}:L^\infty_\mu\to({\rm ran}\,\ms P)''$ where $\Gamma_{\ms P}(f)=\int_X f\,d\ms P$ for all $f\in L^\infty_\mu$ and $({\rm ran}\,\ms P)''\subset\mc L(\mc K_1)$ is the double commutant of the range ${\rm ran}\,\ms P:=\{\ms P(A)\,|\,A\in\mc A\}$ of $\ms P$ and, the similar extension $\Gamma_{\ms Q}$ of $\ms Q$ mediates a normal *-isomorphism between $L^\infty_\nu$ and $({\rm ran}\,\ms Q)''$. Let us denote by $\pi:({\rm ran}\,\ms P)''\vN({\rm ran}\,\ms Q)''\to L^\infty_\mu\vN L^\infty_\nu\simeq L^\infty_{\mu\times\nu}$ the inverse of $\Gamma_{\ms M}\otimes\Gamma_{\ms N}$; as a *-isomorphism between von Neumann algebras, $\pi$ is normal. Above, we have denoted by $\mf M\vN\mf N$ the von Neumann tensor product of von Neumann algebras $\mf M$ and $\mf N$. Note that $\ms P\otimes\ms Q\ll\mu\times\nu$ as $\mu\times\nu=\tr{(\sigma_1\otimes\sigma_2)(\ms P\otimes\ms Q)(\cdot)}$ and $\sigma_1\otimes\sigma_2$ is faithful on $\mc K_1\otimes\mc K_2$. Let us assume that $\ms G:\mc A\otimes\ms B\to\mc L(\hil)$ is a POVM such that $\ms G\ll\ms P\otimes\ms Q$. It follows that $\ms G\ll\mu\times\nu$. We may now extend $\ms G$ into the normal positive linear map $\Gamma_{\ms G}:L^\infty_{\mu\times\nu}\to\mc L(\hil)$ defined through $\Gamma_{\ms G}(f\otimes g)=\int_{X\times Y}f(x)g(y)\,d\ms G(x,y)$ for all $f\in L^\infty_\mu$ and $g\in L^\infty_\nu$.

Let $\{\fii_i\}_i\subset\hil$ be an orthonormal basis and $t_i>0$ be non-vanishing numbers such that $\sum_i t_i=1$. Let us define $\rho_0:=\sum_i t_i|\fii_i\>\<\fii_i|\in\mc S(\hil)$ and $\Omega_0:=\sum_i\sqrt{t_i}\fii_i\otimes\fii_i\in\hil\otimes\hil$. It follows that $\rho_0$ is faithful. Denote $\mc L(\hil)\vN({\rm ran}\,\ms P)''\vN({\rm ran}\,\ms Q)''$ by $\mf M$. Using the fact that $\Gamma_{\ms G}\circ\pi$ is normal and unital, we may define the normal positive functional $\rho$ on $\mf M$ through $\rho(S)=\<\Omega_0|[{\rm id}\otimes(\Gamma_{\ms G}\circ\pi)](S)\Omega_0\>$ for all $S\in\mf M$; here ${\rm id}$ is the identity map on $\mc L(\hil)$. Define the transpose $\mc L(\hil)\ni R\mapsto R^T\in\mc L(\hil)$ w.r.t. the basis $\{\fii_i\}_i$, i.e.,\ $\<\fii_i|R\fii_j\>=\<\fii_j|R^T\fii_i\>$ for all $R\in\mc L(\hil)$ and $i,\,j$. Since $(\Gamma_{\ms G}\circ\pi)\big(\Gamma_{\ms P}(f)\otimes\Gamma_{\ms Q}(g)\big)=\Gamma_{\ms G}(f\otimes g)$ for all $f\in L^\infty_\mu$ and $g\in L^\infty_\nu$, we have
\begin{align*}
\rho\big(R\otimes\Gamma_{\ms P}(f)\otimes\Gamma_{\ms Q}(g)\big)=&\big\<\Omega_0\big|\big(R\otimes\Gamma_{\ms G}(f\otimes g)\big)\Omega_0\big\>=\sum_{i,j}\sqrt{t_it_j}\<\fii_i|R\fii_j\>\<\fii_i|\Gamma_{\ms G}(f\otimes g)\fii_j\>\\
=&\sum_{i,j}\sqrt{t_it_j}\<\fii_j|R^T\fii_i\>\<\fii_i|\Gamma_{\ms G}(f\otimes g)\fii_j\>=\tr{\rho_0^{1/2}R^T\rho_0^{1/2}\Gamma_{\ms G}(f\otimes g)}
\end{align*}
for all $R\in\mc L(\hil)$, $f\in L^\infty_\mu$, and $g\in L^\infty_\nu$.

According to Theorem 1 of Chapter 3 and Theorem 1 of Chapter 4 of Part I of \cite{Dixmier}, there is a state $\tau\in\mc S(\hil\otimes\mc K_1\otimes\mc K_2)$ (which is not typically unique) such that $\rho(S)=\tr{\tau S}$ for all $S\in\mf M$. According to our earlier calculation, we have $\tr{\tau(R\otimes\id_{\mc K_1\otimes\mc K_2})}=\rho(R\otimes\id_{\mc K_1\otimes\mc K_2})=\tr{\rho_0^{1/2}R^T\rho_0^{1/2}}=\tr{\rho_0 R}$ for all $R\in\mc L(\hil)$, implying that the reduced state of $\tau$ on $\hil$ is the faithful $\rho_0$. According to Theorem 1 and the beginning of Section 3.1 of \cite{Haapasalo2019a}, there is a unique channel $\Phi:\mc T(\hil)\to\mc T(\mc K_1\otimes\mc K_2)$ such that $\tr{\tau(R\otimes R')}=\tr{\Phi(\rho_0^{1/2}R^T\rho_0^{1/2})R'}$ for all $R\in\mc L(\hil)$ and $R'\in\mc L(\mc K_1\otimes\mc K_2)$. Thus, for all $R\in\mc L(\hil)$, $A\in\mc A$, and $B\in\mc B$,
\begin{align*}
\tr{\Phi(\rho_0^{1/2}R^T\rho_0^{1/2})\big(\ms P(A)\otimes\ms Q(B)\big)}=&\tr{\tau\big(R\otimes\ms P(A)\otimes\ms Q(B)\big)}\\
=&\rho\big(R\otimes\ms P(A)\otimes\ms Q(B)\big)=\tr{\rho_0^{1/2}R^T\rho_0^{1/2}\ms G(A\times B)},
\end{align*}
where we have used the fact that $\ms P(A)=\Gamma_{\ms P}(\chi_A)$, $\ms Q(B)=\Gamma_{\ms Q}(\chi_B)$, and $\ms G(A\times B)=\Gamma_{\ms G}(\chi_A\otimes\chi_B)$ where $\chi_C$ is the characteristic function of the set $C$. Since operators of the form $\rho_0^{1/2}R^T\rho_0^{1/2}$ with $R\in\mc L(\hil)$ constitute a trace-norm-dense subset of $\mc T(\hil)$ (a fact which is easily verified), we now have $\tr{\Phi(\rho)\big(\ms P(A)\otimes\ms Q(B)\big)}=\tr{\rho\ms G(A\times B)}$ for all $\rho\in\mc T(\hil)$, $A\in\mc A$, and $B\in\mc B$, implying that $\Phi^*\big(\ms P(A)\otimes\ms Q(B)\big)=\ms G(A\times B)$ for all $A\in\mc A$ and $B\in\mc B$.
\end{proof}

As already mentioned, there are some cases where joint measurements through broadcasting are bound to fail. The absolute continuity requirement in the preceding theorem is crucial in barring these cases off. To better understand the problematic cases, let us make things formal. Fix a Hilbert space $\hil$, measurable spaces $(X,\ms A)$ and $(Y,\mc B)$, and PVMs $\ms P:\mc A\to\mc L(\hil)$ and $\ms Q:\mc B\to\mc L(\hil)$. For the joint measurability of $\ms P$ and $\ms Q$, it is necessary that $\ms P$ and $\ms Q$ commute, i.e.,\ $\ms P(A)\ms Q(B)=\ms Q(B)\ms P(A)$ for all $A\in\mc A$ and $B\in\mc B$ \cite{HeReSt2008}. This is also a sufficient condition for joint measurability if the value spaces $(X,\mc A)$ and $(Y,\mc B)$ are standard Borel. Whenever $\ms P$ and $\ms Q$ are jointly measurable their joint POVM $\ms G$ is unique and given by $\ms G(A\times B)=\ms P(A)\ms Q(B)$ for all $A\in\mc A$ and $B\in\mc B$ \cite{HaHePe2014}. We call this as the {\it canonical joint measurement of $\ms P$ and $\ms Q$}. As an example consider the canonical PVM $\ms Q:\mc B(\R)\to\mc L\big(L^2(\R)\big)$ defined by $\big(\ms Q(A)\fii\big)(x)=\chi_A(x)\fii(x)$ for all $A\in\mc B(\R)$, $\fii\in L^2(\R)$, and $x\in\R$. The canonical joint PVM for $\ms Q$ with itself, i.e.,\ the product defined by $\ms G(A\times B)=\ms Q(A)\ms Q(B)=\ms Q(A\cap B)$ for all $A,\,B\in\mc B(\R)$ is supported by the diagonal $\{(x,x)\,|\,x\in\R\}$ and clearly not absolutely continuous w.r.t. $\ms Q\otimes\ms Q$. Thus, the conditions for Theorem \ref{theor:PVMbroadcast} do not hold, and the following result, which is largely a consequence of Proposition 3.6 of \cite{KaLuLu2015}, tells us that, e.g.,\ in this case, joint measurement through broadcasting will fail.

\begin{proposition}
Let $\hil$ be a Hilbert space, $(X,\mc A)$ be a measurable space, and $\ms Q:\mc A\to\mc L(\hil)$ be a PVM. The canonical joint measurement of $\ms Q$ with itself can be realized by broadcasting, i.e.,\ there is a channel $\Phi:\mc T(\hil)\to\mc T(\hil\otimes\hil)$ such that $\Phi^*\big(\ms Q(A)\otimes\ms Q(B)\big)=\ms Q(A)\ms Q(B)=\ms Q(A\cap B)$ for all $A,\,B\in\mc A$ if and only if there is a discrete PVM $(Q_i)_{i\in I}\subseteq\mc L(\hil)$ such that, for all $A\in\mc A$,
\begin{equation}\label{eq:Qpp}
\ms Q(A)=\sum_{i\in I}\delta_i(A)Q_i
\end{equation}
where $\delta_i:\mc A\to\{0,1\}$ is a probability measure for all $i\in I$.
\end{proposition}

Before going into the proof of the above claim, let us note that the sequence of maps $\delta_i:\mc A\to\{0,1\}$ above constitutes a Markov kernel which, according to Equation \eqref{eq:Qpp} transforms the discrete PVM $(Q_i)_{i\in I}$ into $\ms Q$. This means that the canonical joint measurement of $\ms Q$ with itself can be realized through broadcasting if and only if $\ms Q$ can be post-processed from a discrete PVM, in which case we could say that $\ms Q$ is essentially discrete.

\begin{proof}
Assume first that the canonical joint measurement of $\ms Q$ with itself can be realized through broadcasting and let $\Phi:\mc T(\hil)\to\mc T(\hil\otimes\hil)$ be the associated broadcasting channel. Denote $\mf A:=({\rm ran}\,\ms Q)''$. Since $\Phi^*$ is normal, we have $\Phi^*(R\otimes S)=RS$ for all $R,\,S\in\mf A$. This means that, by restriction, we have a normal positive unital map $\Phi^*:\mf A\vN\mf A\to\mf A$ such that $\Phi^*(R\otimes S)=RS$ for all $R,\,S\in\mf A$. Proposition 3.6 of \cite{KaLuLu2015} now tells us that $\mf A$ is generated by minimal projections, i.e.,\ there is an at most countably infinite set $\{Q_i\}_{i\in I}\subset\mf A$ of mutually orthogonal projections which generate $\mf A$. Thus, $\sum_{i\in I}Q_i=\id_\hil$. Since the orthogonal set $\{Q_i\}_{i\in I}$ generates $\mf A$, for all $A\in\mc A$, there are unique $\delta_i(A)\in\C$ ($i\in I$) such that $\ms Q(A)=\sum_{i\in I}\delta_i(A)Q_i$. Using $Q(A)^2=\ms Q(A)^*=\ms Q(A)$ for all $A\in\mc A$ and the orthogonality of $\{Q_i\}_{i\in I}$, we easily have $\delta_i(A)\in\{0,1\}$ for all $i\in I$ and $A\in\mc A$. Similarly, using $\ms Q(X)=\id_\hil$ and $\ms Q(\emptyset)=0$, we have $\delta_i(X)=1$ and $\delta_i(\emptyset)=0$ for all $i\in I$. Let now $A_1,\,A_2,\ldots\in\mc A$ be a disjoint sequence. We have
\begin{align*}
\sum_{i\in I}\delta_i\big(\cup_{k=1}^\infty A_k\big)Q_i=&\ms Q\big(\cup_{k=1}^\infty A_k\big)=\sum_{k=1}^\infty\ms Q(A_k)=\sum_{k=1}^\infty\sum_{i\in I}\delta_i(A_k)Q_i=\sum_{i\in I}\left(\sum_{k=1}^\infty\delta_i(A_k)\right)Q_i,
\end{align*}
which, together with the orthogonality of $\{Q_i\}_{i\in I}$, implies that $\delta_i\big(\cup_{k=1}^\infty A_k\big)=\sum_{k=1}^\infty\delta_i(A_k)$ for all $i\in I$. Thus, $\delta_i:\mc A\to\{0,1\}$ is a probability measure for all $i\in I$.

Let now $\ms Q$ be such that there is an orthogonal set $\{Q_i\}_{i\in I}\subset\mc L(\hil)$ of projections where $I$ is at most countably infinite so that $\ms Q$ can be post-processed with the Markov kernel $(\delta_i)_{i\in I}$ from the discrete PVM $(Q_i)_{i\in I}$. For each $i\in I$, let $\{e_{i,k}\}_{k\in K_i}$ be an orthonormal basis for the support of $Q_i$. Thus, $\bigcup_{i\in I}\{e_{i,k}\}_{k\in K_i}$ is an orthonormal basis for $\hil$. Let us define the channel $\Phi:\mc T(\hil)\to\mc T(\hil\otimes\hil)$ through
$$
\Phi(\rho)=\sum_{i\in I}\sum_{k\in K_i}\<e_{i,k}|\rho e_{i,k}\>|e_{i,k}\otimes e_{i,k}\>\<e_{i,k}\otimes e_{i,k}|,\qquad\rho\in\mc T(\hil).
$$
It easily follows that the dual $\Phi^*$ can be defined through
$$
\Phi^*(R\otimes S)=\sum_{i\in I}\sum_{k\in K_i}\<e_{i,k}|Re_{i,k}\>\<e_{i,k}|Se_{i,k}\>|e_{i,k}\>\<e_{i,k}|,\qquad R,\,S\in\mc L(\hil).
$$
Since, for all $i\in I$ and $k\in K_i$, $P_{i,k}:=|e_{i,k}\>\<e_{i,k}|$ commutes with $Q_j$ for all $j\in I$ and $\{Q_n\}_{n\in I}$ generates the range of $\ms Q$, we find that $|e_{i,k}\>\<e_{i,k}|$ commutes with $\ms Q(A)$ for all $i\in I$, $k\in K_i$, and $A\in\mc A$. Thus, for all $A,\,B\in\mc A$,
\begin{align*}
\Phi^*\big(\ms Q(A)\otimes\ms Q(B)\big)=&\sum_{i\in I}\sum_{k\in K_i}\<e_{i,k}|\ms Q(A)e_{i,k}\>\<e_{i,k}|\ms Q(B)e_{i,k}\>|e_{i,k}\>\<e_{i,k}|\\
=&\sum_{i\in I}\sum_{k\in K_i}P_{i,k}\ms Q(A)P_{i,k}\ms Q(B)P_{i,k}=\sum_{i\in I}\sum_{k\in K_i}P_{i,k}\ms Q(A\cap B)P_{i,k}\\
=&\ms Q(A\cap B)\sum_{i\in I}\sum_{k\in K_i}P_{i,k}=\ms Q(A\cap B).
\end{align*}
This shows that the canonical joint measurement of $\ms Q$ with itself can be realized by broadcasting mediated by the channel $\Phi$.
\end{proof}

\begin{remark}\label{rem:measandprep}
Since the conditions of Theorem \ref{theor:PVMbroadcast} essentially always hold for discrete PVMs, we may always realize (approximate) joint measurements of discrete PVMs via broadcasting. The only requirement is that the (approximate) joint POVM be supported on the product of the supports of the PVMs. In fact, measure-and-prepare channels are enough for this task. To see this, let $\ms P=(P_x)_{x\in X}$ and $\ms Q=(Q_y)_{y\in Y}$ be PVMs in a Hilbert space $\hil$. Let us assume that $X$ is the support of $\ms P$, i.e.,\ $P_x\neq0$ for all $x\in X$, and that $Y$ is the support of $\ms Q$; we naturally have this liberty. Let $\ms G=(G_{x,y})_{(x,y)\in X\times Y}$ be any POVM in $\hil$. To construct the measure-and-prepare broadcasting channel, let us pick states $\sigma^{(1)}_x\in\mc S(\hil)$ and $\sigma^{(2)}_y\in\mc S(\hil)$ such that $\sigma^{(1)}_x\leq P_x$ for all $x\in X$ and $\sigma^{(2)}_y\leq Q_y$ for all $y\in Y$. We may now define the channel $\Phi:\mc T(\hil)\to\mc T(\hil\otimes\hil)$ through
$$
\Phi(\rho)=\sum_{x\in X}\sum_{y\in Y}\tr{\rho G_{x,y}}\,\sigma^{(1)}_x\otimes\sigma^{(2)}_y,\qquad\rho\in\mc S(\hil).
$$
Using the fact that $\tr{\sigma^{(1)}_xP_{x'}\otimes\sigma^{(2)}_yQ_{y'}}=\delta_{x,x'}\delta_{y,y'}$ (where $\delta_{a,b}$ are the Kronecker symbols), it can be easily shown that $\Phi^*(P_x\otimes Q_y)=G_{x,y}$ for all $x\in X$ and $y\in Y$. Moreover, $\Phi$ can be interpreted as a measure-and-prepare channel as it can be realized as a process where we first measure $\ms G$ and, depending on the outcome $(x,y)\in X\times Y$ of this measurement, we prepare the state $\sigma_{x,y}:=\sigma^{(1)}_x\otimes\sigma^{(2)}_y$. Naturally, any other post-measurement states such that $\sigma_{x,y}\leq P_x\otimes Q_y$ will also do.
\end{remark}

\section{Conditions for the existence of broadcasting channels for joint measurements for fuzzy POVMs}\label{sec:noisyPOVMs}

We have seen that broadcasting channels generating a joint POVM $\ms G$ for PVMs $\ms P$ and $\ms Q$ exit given that the minimal condition $\ms G\ll\ms P\otimes\ms Q$ is satisfied. However, when the above PVMs $\ms P$ and $\ms Q$ are replaced by more general POVMs, the situation is more complicated. However, we can provide some necessary and sufficient conditions which greatly simplify in some special cases. In this section, we concentrate on discrete POVMs in a finite-dimensional Hilbert space. Although $\mc L(\hil)$ and $\mc T(\hil)$, for a finite-dimensional Hilbert space $\hil$, are, as sets, identical, we continue to use these differentiating notations since their physical interpretations are completely different. First, using the finite-dimensional version of Arveson's extension theorem, we can give a general characterization for the existence of a broadcasting channel for a pair of POVMs and their intended joint POVM.

\begin{theorem}\label{theor:hankalaehto}
Suppose that $\hil$ is a $D$-dimensional Hilbert space ($D<\infty$), $X,\,Y\neq\emptyset$ are finite sets, and $\ms M=(M_x)_{x\in X}$, $\ms N=(N_y)_{y\in Y}$, and $\ms G=(G_{x,y})_{(x,y)\in X\times Y}$ are POVMs. The POVM $\ms G$ can be generated from $\ms M$ and $\ms N$ through broadcasting if and only if the implications
$$
\begin{array}{rrcl}
\alpha_{x,y}\in\C:&\sum_{x\in X}\sum_{y\in Y}\alpha_{x,y}M_x\otimes N_y=0&\Rightarrow&\sum_{x\in X}\sum_{y\in Y}\alpha_{x,y}G_{x,y}=0,\\
H_{x,y}\in\mc M_{D^2}(\C):&\sum_{x\in X}\sum_{y\in Y}H_{x,y}\otimes M_x\otimes N_y\geq0&\Rightarrow&\sum_{x\in X}\sum_{y\in Y}H_{x,y}\otimes G_{x,y}\geq0
\end{array}
$$
hold.
\end{theorem}

\begin{proof}
Let us denote by $\mc L$ the linear hull of $\{M_x\otimes N_y\}_{(x,y)\in X\times Y}$. Let us first assume that the first implication of the claim holds and na\"{\i}vely define a linear map $\Psi_0:\mc L\to\mc L(\hil)$ through $\Psi_0(M_x\otimes N_y)=G_{x,y}$ for all $x\in X$ and $y\in Y$; using the first implication of the claim, we now show that this linear map is, indeed, well defined. Assume that $L\in\mc L$ can be given the following forms
$$
\sum_{x\in X}\sum_{y\in Y}\alpha^+_{x,y}M_x\otimes N_y=L=\sum_{x\in X}\sum_{y\in Y}\alpha^-_{x,y}M_x\otimes N_y
$$
where $\alpha^\pm_{x,y}\in\C$ for all $x\in X$ and $y\in Y$. This means that
$$
\begin{array}{rrcl}
&\sum_{x\in X}\sum_{y\in Y}(\alpha^+_{x,y}-\alpha^-_{x,y})M_x\otimes N_y=0&\Rightarrow&\sum_{x\in X}\sum_{y\in Y}(\alpha^+_{x,y}-\alpha^-_{x,y})G_{x,y}=0\\
\Leftrightarrow&\sum_{x\in X}\sum_{y\in Y}\alpha^+_{x,y}G_{x,y}&=&\sum_{x\in X}\sum_{y\in Y}\alpha^-_{x,y}G_{x,y}\\
\Leftrightarrow&\sum_{x\in X}\sum_{y\in Y}\alpha^+_{x,y}\Psi_0(M_x\otimes N_y)&=&\sum_{x\in X}\sum_{y\in Y}\alpha^-_{x,y}\Psi_0(M_x\otimes N_y),
\end{array}
$$
implying that $\Psi_0$ is well defined.

Before applying Arveson's extension theorem, we have to find the necessary and sufficient conditions under which $\Psi_0$ is completely positive on the operator system $\mc L$. Due to Choi's theorem, $\Psi_0$ is completely positive if (and only if) it is $D^2$-positive. We can thus concentrate on characterizing the $D^2$-positivity of $\Psi_0$. Since every element $L\in\mc M_{D^2}(\C)\otimes\mc L$ can be written in the form $L=\sum_{x\in X}\sum_{y\in Y}H_{x,y}\otimes M_x\otimes N_y$ for some $H_{x,y}\in\mc M_{D^2}(\C)$ ($x\in X$, $y\in Y$), we find that complete positivity of $\Psi_0$ is equivalent with the implication of the claim. Indeed, whenever $L\in\mc M_{D^2}(\C)\otimes\mc L$ is like that above, we have
$$
({\rm id}_{D^2}\otimes\Psi_0)(L)=\sum_{x\in X}\sum_{y\in Y}H_{x,y}\otimes\Psi_0(M_x\otimes N_y)=\sum_{x\in X}\sum_{y\in Y}H_{x,y}\otimes G_{x,y}
$$
where ${\rm id}_{D^2}$ is the identity map within the matrix algebra $\mc M_{D^2}(\C)$. Whenever this complete positivity condition holds, Arveson's extension theorem gives us a completely positive linear unital map $\Psi:\mc L(\hil\otimes\hil)\to\mc L(\hil)$ such that $\Psi|_{\mc L}=\Psi_0$, i.e.,\ $\Psi(M_x\otimes N_y)=G_{x,y}$ for all $x\in X$ and $y\in Y$. Since $\hil$ is finite dimensional (whence $\Psi$ is normal) there is a channel $\Phi:\mc T(\hil)\to\mc T(\hil\otimes\hil)$ such that $\Psi=\Phi^*$ and the `if' part of the claim is proven.

On the other hand, if the broadcasting channel $\Phi$ of the claim exists, its dual $\Phi^*$ must be a completely positive linear map on $\mc L$ such that $\Phi^*(M_x\otimes N_y)=G_{x,y}$ for all $x\in X$ and $y\in Y$. The linearity of $\Phi^*$ easily implies that the first implication of the claim holds. Moreover, the complete positivity of $\Phi^*|_{\mc L}$ exactly corresponds to the second implication stated in the claim.
\end{proof}

Let us note that the first implication from the conditions for broadcastability above can be lifted if $\ms M$ and $\ms N$ are linearly independent, i.e.,\ the sets $\{M_x\}_{x\in X}$ and $\{N_y\}_{y\in Y}$ are linearly independent. This is relevant as the class of linearly independent (discrete) POVMs contains several important POVMs. Especially, extreme POVMs are linearly independent. We say that a POVM $\ms M:\mc A\to\mc L(\hil)$ (where $(X,\mc A)$ is a measurable space and $\hil$ is a (not necessarily finite-dimensional) Hilbert space) is extreme if it is a convex extreme point of the convex set of POVMs in $\hil$ with the value space $(X,\mc A)$. The extreme POVMs can, in general, be defined by their minimal Na\u{\i}mark dilations \cite{Pellonpaa2011}; also equivalent characterizations for discrete POVMs in finite dimensions exist \cite{DaLoPe2005,Parthasarathy99}. In particular, these conditions imply that, if a discrete POVM is extreme, then that POVM is linearly independent. If a discrete POVM is of rank 1, it is extreme if and only if it is linearly independent.

\begin{remark}\label{rem:positivitycond}
Let us denote the dimension of $\hil$ in Theorem \ref{theor:hankalaehto} by $D$. We should note that, in the claim of Theorem \ref{theor:hankalaehto}, the number $n$ can be set as $D^2$. This is because a channel $\Phi:\mc T(\hil)\to\mc T(\hil\otimes\hil)$ is completely positive if and only if it is $D^2$-positive. Moreover, we can recast the second implication in the statement of Theorem \ref{theor:hankalaehto} in the form
\begin{equation}\label{eq:secondform}
H_{x,y}\in\mc M_{D^2}(\C):\quad\sum_{x\in X}\sum_{y\in Y}\tr{H_{x,y}(M_x\otimes N_y)}\geq0\ \Rightarrow\ \sum_{x\in X}\sum_{y\in Y}\tr{H_{x,y}(G_{x,y}\otimes\id_\hil)}\geq0
\end{equation}
where we identify $\mc M_{D^2}(\C)$ with $\mc L(\hil\otimes\hil)$.

To prove this, let us define $\Omega:=\sum_{i,j=1}^D\fii_i\otimes\fii_j\otimes\fii_i\otimes\fii_j$ where $\{\fii_i\}_{i=1}^D$ is a fixed orthonormal basis of $\hil$. It is easy to show that any $w\in\hil^{\otimes 4}$ can be written in the form $w=(K\otimes\id_{\hil\otimes\hil})\Omega$ for some $K\in\mc L(\hil\otimes\hil)\simeq\mc M_{D^2}(\C)$. Fix $H_{x,y}\in\mc M_{D^2}(\C)$ ($x\in X$, $y\in Y$). Clearly $\sum_{x\in X}\sum_{y\in Y}H_{x,y}\otimes G_{x,y}\geq0$ if and only if $\sum_{x\in X}\sum_{y\in Y}H_{x,y}\otimes G_{x,y}\otimes\id_\hil\geq0$, i.e.,
\begin{align*}
0\leq&\sum_{x\in X}\sum_{y\in Y}\<(K\otimes\id_{\hil\otimes\hil})\Omega|(H_{x,y}\otimes G_{x,y}\otimes\id_\hil)(K\otimes\id_{\hil\otimes\hil})\Omega\>\\
=&\sum_{x\in X}\sum_{y\in Y}\<\Omega|(K^*H_{x,y}K\otimes G_{x,y}\otimes\id_\hil)\Omega\>=\sum_{x\in X}\sum_{y\in Y}\tr{K^TH_{x,y}^TK^{*\,T}(G_{x,y}\otimes\id_\hil)}
\end{align*}
for all $K\in\mc L(\hil\otimes\hil)$ where the final form is obtained through direct calculation and $\mc L(\hil\otimes\hil)\ni R\mapsto R^T\in\mc L(\hil\otimes\hil)$ is the transpose defined by the basis $\{\fii_i\otimes\fii_j\}_{i,j=1}^D$. Similarly, $\sum_{x\in X}\sum_{y\in Y}H_{x,y}\otimes M_x\otimes N_y\geq0$ if and only if $\sum_{x\in X}\sum_{y\in Y}\tr{K^TH_{x,y}^TK^{*\,T}(M_x\otimes N_y)}\geq0$. Since $\sum_{x\in X}\sum_{y\in Y}K^*H_{x,y}K\otimes M_x\otimes N_y\geq0$ whenever $H_{x,y}\in\mc M_{D^2}(\C)$ are such that $\sum_{x\in X}\sum_{y\in Y}H_{x,y}\otimes M_x\otimes N_y\geq0$ and $K\in\mc M_{D^2}(\C)$ is any matrix, we now see that the implication \eqref{eq:secondform} holds upon assuming the second implication of Theorem \ref{theor:hankalaehto}.
\end{remark}

Using Theorem \ref{theor:hankalaehto}, we may derive easily checked necessary and sufficient conditions for (approximate) joint POVMs under which this POVM can be generated through broadcasting from POVMs which are convex combinations of PVMs and (coloured) noise.

\begin{corollary}\label{cor:fuzzyPVM}
Let $\hil$ be a finite-dimensional Hilbert space, $\ms P=(P_x)_{x\in X}$ and $\ms Q=(Q_y)_{y\in Y}$ be PVMs in $\hil$, and $(p_x)_{x\in X}$ and $(q_y)_{y\in Y}$ be probability vectors. Assume that $P_x\neq0\neq Q_y$ for all $(x,y)\in X\times Y$. Define, for any $\lambda,\,\mu\in[0,1]$ the POVMs $\ms P^\lambda=(P^\lambda_x)_{x\in X}$ and $\ms Q^\mu=(Q^\mu_y)_{y\in Y}$ through
$$
P^\lambda_x=\lambda P_x+(1-\lambda)p_x\id_\hil,\qquad Q^\mu_y=\mu Q_y+(1-\mu)q_y\id_\hil
$$
for all $x\in X$ and $y\in Y$. For $\lambda,\,\mu\in(0,1]$, a POVM $\ms G=(G_{x,y})_{(x,y)\in X\times Y}$ can be generated through broadcasting from $\ms P^\lambda$ and $\ms Q^\mu$ if and only if, for all $x\in X$ and $y\in Y$,
\begin{equation}\label{eq:quditcond}
G_{x,y}-(1-\mu)q_y G^{(1)}_x-(1-\lambda)p_x G^{(2)}_y+(1-\lambda)(1-\mu)p_x q_y\id_\hil\geq0
\end{equation}
where $\ms G^{(1)}=(G^{(1)}_x)_{x\in X}$ and $\ms G^{(2)}=(G^{(2)}_y)_{y\in Y}$ are the margins of $\ms G$, i.e., $G^{(1)}_x=\sum_{y'\in Y}G_{x,y'}$ and $G^{(2)}_y=\sum_{x'\in X}G_{x',y}$ for all $x\in X$ and $y\in Y$.
\end{corollary}

Before proving this claim, note that, in the case $\lambda=1=\mu$, the condition \eqref{eq:quditcond} is trivial. This is in line with Theorem \ref{theor:PVMbroadcast} since the condition of this theorem is satisfied in the above situation.

\begin{proof}
Pick $\lambda,\,\mu\in(0,1]$ and a POVM $\ms G=(G_{x,y})_{(x,y)\in X\times Y}$. Note that, in our situation, the set $\{P^\lambda_x\otimes Q^\mu_y\}_{(x,y)\in X\times Y}$ is linearly independent, so we can concentrate on the latter implication of Theorem \ref{theor:hankalaehto}. Let $D:={\rm dim}\,\hil$. We first characterize those $H_{x,y}\in\mc M_{D^2}(\C)$ ($x\in X$, $y\in Y$) such that $\sum_{x\in X}\sum_{y\in Y}H_{x,y}\otimes P^\lambda_x\otimes Q^\mu_y\geq0$ and then use the methods introduced in Remark \ref{rem:positivitycond} to derive the condition of the claim.

Let $H_{x,y}\in\mc M_{D^2}(\C)$ ($x\in X$, $y\in Y$) be such that $\sum_{x\in X}\sum_{y\in Y}H_{x,y}\otimes P^\lambda_x\otimes Q^\mu_y\geq0$. Using the fact that $P^\lambda_x=\sum_{x'\in X}\big(\lambda\delta_{x,x'}+(1-\lambda)p_x\big)P_{x'}$ and $Q^\mu_y=\sum_{y'\in Y}\big(\mu\delta_{y,y'}+(1-\mu)q_y\big)Q_{y'}$ for all $x\in X$ and $y\in Y$, we have
\begin{align}
0\leq&\sum_{x\in X}\sum_{y\in Y}H_{x,y}\otimes P^\lambda_x\otimes Q^\mu_y\nonumber\\
=&\sum_{x,x'\in X}\sum_{y,y'\in Y}\big(\lambda\delta_{x,x'}+(1-\lambda)p_x\big)\big(\mu\delta_{y,y'}+(1-\mu)q_y\big)H_{x,y}\otimes P_{x'}\otimes Q_{y'}\nonumber\\
=&\lambda\mu\sum_{x\in X}\sum_{y\in Y}K_{x,y}\otimes P_x\otimes Q_y\label{eq:posehto}
\end{align}
where $K_{x,y}:=(\lambda\mu)^{-1}\sum_{x'\in X}\sum_{y'\in Y}\big(\lambda\delta_{x,x'}+(1-\lambda)p_{x'}\big)\big(\mu\delta_{y,y'}+(1-\mu)q_{y'}\big)H_{x',y'}$ for all $x\in X$ and $y\in Y$. Since $\ms P\otimes\ms Q$ is non-vanishing on $X\times Y$, Inequality \eqref{eq:posehto} is equivalent with $K_{x,y}\geq0$ for all $x\in X$ and $y\in Y$. Simple matrix inversion yields, for all $x\in X$ and $y\in Y$,
\begin{equation}\label{eq:KstaHksi}
H_{x,y}=\sum_{x'\in X}\sum_{y'\in Y}\left(\delta_{x,x'}-(1-\lambda)p_{x'}\right)\left(\delta_{y,y'}-(1-\mu)q_{y'}\right)K_{x',y'}.
\end{equation}
Thus, $H_{x,y}$ are such that $\sum_{x\in X}\sum_{y\in Y}H_{x,y}\otimes P^\lambda_x\otimes Q^\mu_y\geq0$ if and only if there are positive $K_{x,y}$ such that Equation \eqref{eq:KstaHksi} holds.

Let $K_{x,y}\in\mc M_{D^2}(\C)$ ($x\in X$, $y\in Y$) be positive matrices and define $H_{x,y}$ through Equation \eqref{eq:KstaHksi}. According to Remark \ref{rem:positivitycond}, $\sum_{x\in X}\sum_{y\in Y}H_{x,y}\otimes G_{x,y}\geq0$ if and only if
\begin{align*}
0\leq&\sum_{x\in X}\sum_{y\in Y}\tr{K^TH_{x,y}^TK^{*\,T}(G_{x,y}\otimes\id_\hil)}\\
=&\sum_{x,x'\in X}\sum_{y,y'\in Y}\left(\delta_{x,x'}-(1-\lambda)p_x\right)\left(\delta_{y,y'}-(1-\mu)q_y\right)\tr{K^TK_{x,y}^TK^{*\,T}(G_{x',y'}\otimes\id_\hil)}
\end{align*}
for all $K\in\mc M_{D^2}(\C)$. As $K^TK_{x,y}^TK^{*\,T}\geq0$ if $K_{x,y}\geq0$, we see that the second implication of Theorem \ref{theor:hankalaehto} holds (and, thus, the channel $\Phi$ of the claim exists) if and only if
$$
\sum_{x,x'\in X}\sum_{y,y'\in Y}\left(\delta_{x,x'}-(1-\lambda)p_x\right)\left(\delta_{y,y'}-(1-\mu)q_y\right)\tr{K_{x,y}(G_{x',y'}\otimes\id_\hil)}\geq0
$$
for all positive $K_{x,y}\in\mc M_{D^2}(\C)$ ($x\in X$, $y\in Y$). This is equivalent with
$$
\begin{array}{rrcl}
&\sum_{x'\in X}\sum_{y'\in Y}\left(\delta_{x,x'}-(1-\lambda)p_x\right)\left(\delta_{y,y'}-(1-\mu)q_y\right)\,G_{x',y'}\otimes\id_\hil&\geq&0\\
\Leftrightarrow&\sum_{x'\in X}\sum_{y'\in Y}\left(\delta_{x,x'}-(1-\lambda)p_x\right)\left(\delta_{y,y'}-(1-\mu)q_y\right)\,G_{x',y'}&\geq&0
\end{array}
$$
for all $x\in X$ and $y\in Y$. The final inequality above is easily seen to coincide with \eqref{eq:quditcond}, proving the claim.
\end{proof}

Let us study the consequences of Corollary \ref{cor:fuzzyPVM} for qubit observables. An unbiased qubit observable $\ms S=(S_+,S_-)$ is a binary POVM in a two-dimensional Hilbert space determined by a vector $\vec{a}\in\R^3$ such that $\|\vec{a}\|\leq1$ through $S_\pm=(1/2)(\id_2\pm\vec{a}\cdot\vec{\sigma})$ where $\vec{a}\cdot\vec{\sigma}:=a_x\sigma_x+a_y\sigma_y+a_z\sigma_z$ where, in turn, $\vec{a}=(a_x,a_y,a_z)$ and $\sigma_x$, $\sigma_y$, and $\sigma_z$ are the Pauli spin matrices. The {\it sharpness} of the above POVM $\ms S$ is $\sharp(\ms S):=\|\vec{a}\|$; $\ms S$ is sharp if and only if $\sharp(\ms S)=\|\vec{a}\|=1$. An unbiased qubit observable $\ms S=(S_+,S_-)$ with $S_\pm=(1/2)(\id_2\pm\vec{a}\cdot\vec{\sigma})$, $\sharp(\ms S)=\|\vec{a}\|=:s>0$ can be written as the convex combination $\ms S=s\ms Q+(1-s)\ms T$ where $\ms Q=(Q_+,Q_-)$ is the PVM with $Q_\pm=(1/2)(\id_2\pm s^{-1}\vec{a}\cdot\vec{\sigma})$ and $\ms T=\big((1/2)\id_2,(1/2)\id_2\big)$, i.e.,\ $S_\pm=sQ_\pm+(1-s)(1/2)\id_2$. This means that unbiased qubit observables are of the form as those noisy POVMs studied in Corollary \ref{cor:fuzzyPVM}.

Let $\ms R$ and $\ms S$ be unbiased qubit POVMs with $\sharp(\ms R)=:r>0$ and $\sharp(\ms S)=:s>0$. For any qubit POVM $\ms G=(G_{\sigma,\tau})_{\sigma,\tau\in\{+,-\}}$, the condition \eqref{eq:quditcond} (with $\lambda=r$ and $\mu=s$) easily becomes
\begin{equation}\label{eq:qubitGehto}
(r+s)G_{\sigma,\tau}-(1-s)G_{\sigma,-\tau}-(1-r)G_{-\sigma,\tau}+\frac{1}{2}(1-r)(1-s)\id_2\geq0
\end{equation}
for all $\sigma,\,\tau\in\{+,-\}$. Let us write $\ms G$ in the form $G_{\sigma,\tau}=a_{\sigma,\tau}(\id_2+\vec{a}_{\sigma,\tau}\cdot\vec{\sigma})$ where $a_{\sigma,\tau}\geq0$ are such that $\sum_{\sigma,\tau}a_{\sigma,\tau}=1$ and $\vec{a}_{\sigma,\tau}\in\R^3$ are such that $\|\vec{a}_{\sigma,\tau}\|\leq1$ and $\sum_{\sigma,\tau}a_{\sigma,\tau}\vec{a}_{\sigma,\tau}=0$; these conditions are easily seen to be the necessary and sufficient conditions for $\ms G$ to be a POVM. Using the fact that $a\id_2+\vec{a}\cdot\vec{\sigma}\geq0$ if and only if $\|\vec{a}\|\leq a$, Equation \eqref{eq:qubitGehto} is equivalent with
\begin{align}
&\|(r+s)a_{\sigma,\tau}\vec{a}_{\sigma,\tau}-(1-s)a_{\sigma,-\tau}\vec{a}_{\sigma,-\tau}-(1-r)a_{-\sigma,\tau}\vec{a}_{-\sigma,\tau}\|\nonumber\\
\leq&(r+s)a_{\sigma,\tau}-(1-s)a_{\sigma,-\tau}-(1-r)a_{-\sigma,\tau}+\frac{1}{2}(1-r)(1-s)\label{eq:qubitgenehto}
\end{align}
for all $\sigma,\,\tau\in\{+,-\}$. Thus, there is a channel $\Phi:\mc T(\C^2)\to\mc T(\C^2\otimes\C^2)$ such that $\Phi^*(R_\sigma\otimes S_\tau)=G_{\sigma,\tau}$ for all $\sigma,\,\tau\in\{+,-\}$ if and only if Inequality \eqref{eq:qubitgenehto} holds for all $\sigma,\,\tau\in\{+,-\}$. If $\ms G$ is unbiased, i.e.,\ $a_{\sigma,\tau}=1/4$ for all $\sigma,\,\tau\in\{+,-\}$, this condition simplifies into
\begin{equation}\label{eq:helpompiehto}
\|(r+s)\vec{a}_{\sigma,\tau}-(1-s)\vec{a}_{\sigma,-\tau}-(1-r)\vec{a}_{-\sigma,\tau}\|\leq2rs
\end{equation}
for all $\sigma,\,\tau\in\{+,-\}$.

Let us look at the special case of an unbiased $\ms G$ where $\vec{a}_{++}=(g/\sqrt{3})(1,1,1)$, $\vec{a}_{+-}=(g/\sqrt{3})(1,-1,-1)$, and $\vec{a}_{-+}=(g/\sqrt{3})(-1,-1,1)$ where $0\leq g\leq1$. This corresponds to an informationally complete $\ms G=:\ms G^g$ whenever $g>0$ and $\ms G^1$ is additionally of rank 1. The informationally completeness of $\ms G^g$ ($g>0$) can be verified easily by showing that the the operators $G^g_{\sigma,\tau}$ linearly span the whole of $\mc L(\C^2)$; this can be done, e.g.,\ by showing that $\id_2$, $\sigma_x$, $\sigma_y$, and $\sigma_z$ can be written as linear combinations of these operators. It is easily verified that the left-hand side of Inequality \eqref{eq:helpompiehto} does not depend on $\sigma,\,\tau$. Indeed, we only have the inequality $g\sqrt{1+r^2+s^2}\leq\sqrt{3}rs$ which can be rewritten as
$$
1+\frac{3}{g^2}\leq\bigg(3\frac{r^2}{g^2}-1\bigg)\bigg(3\frac{s^2}{g^2}-1\bigg).
$$
In the special case $g=1$ (where $\ms G^1$ is a rank-1 informationally complete POVM), this is possible if and only if $r=1=s$, i.e.,\ $\ms R$ and $\ms S$ are sharp. This means that this rank-1 POVM can be generated through broadcasting from unbiased qubit observables if and only if these observables are sharp.

In fact, any binary qubit POVM can be fitted in the framework of Corollary \ref{cor:fuzzyPVM}. Indeed, given a qubit POVM $\ms R=(R_+,R_-)$, we find $a\in[0,1]$ and $\vec{a}\in\R^3$, $\|\vec{a}\|\leq\min\{a,1-a\}$ such that $R_+=a\id_2+\vec{a}\cdot\vec{\sigma}$ and $R_-=(1-a)\id_2-\vec{a}\cdot\vec{\sigma}$. There is a (unique) number $u\in[-1,1]$ such that $a=(1/2)(1+u)$ and, when we make the assumption (which we are free to make) that $a\geq 1-a$, then $u\in[0,1]$ and $\|\vec{a}\|\leq 1-a=(1/2)(1-u)\leq 1/2$. Thus, there is $\lambda\in[0,1]$ such that $\|\vec{a}\|=\lambda/2$. Moreover,
$$
\frac{1}{2}\lambda=\|\vec{a}\|\leq\frac{1}{2}(1-u)\quad\Leftrightarrow\quad u\leq 1-\lambda,
$$
meaning that we may define $p_\pm:=(1/2)(1-\lambda\pm u)/(1-\lambda)\in[0,1]$ with $p_++p_-=1$. Let us define $\hat{a}:=\|\vec{a}\|^{-1}\vec{a}$ whenever $\vec{a}\neq0$ and, if $\vec{a}=0$, we can let $\hat{a}$ be any unit vector in $\R^3$. It now easily follows that $\ms R=\lambda\ms R^\sharp+(1-\lambda)\ms T$ where $\ms R^\sharp=(R^\sharp_+,R^\sharp_-)$ with $R^\sharp_\pm=(1/2)(\id_2\pm\hat{a}\cdot\vec{\sigma})$ and $\ms T=(T_+,T_-)$ with $T_\pm=p_\pm\id_2$. Let $\ms S=(S_+,S_-)$ be another binary qubit POVM with $S_+=b\id_2+\vec{b}\cdot\vec{\sigma}$ with $\|\vec{b}\|\leq\min\{b,1-b\}$. Expressing $\ms S$ as above as a mixture of a binary qubit PVM and (biased) noise and using Corollary \ref{cor:fuzzyPVM}, we find that a qubit POVM $\ms G=(G_{++},G_{+-},G_{-+},G_{--})$ can be generated through broadcasting from $\ms R$ and $\ms S$ if and only if
\begin{align*}
&G_{\sigma,\tau}+\frac{1}{4}\big(1-\tau+2(\tau b-\|\vec{b}\|)\big)\big(1-\sigma+2(\sigma a-\|\vec{a}\|)\big)\id_2\\
\geq&\frac{1}{2}\big(1-\tau+2(\tau b-\|\vec{b}\|)\big)G^{(1)}_\sigma+\frac{1}{2}\big(1-\sigma+2(\sigma a-\|\vec{a}\|)\big)G^{(2)}_\tau
\end{align*}
for all $\sigma,\,\tau\in\{+,-\}$. Conditions for generating joint measurements for qubit POVMs and for non-disturbance in this setting has been studied more in detail in \cite{Heinosaari2016}.

\section{Examples on broadcasting channels for joint measurements}\label{sec:examples}

The results thus far have concentrated on the question of existence of a broadcasting channel for a pair of POVMs and their intended (approximate) joint POVM. We now approach the problem from another angle: how to `engineer' a broadcasting channel for a pair of POVMs so that the POVM resulting from the broadcasting procedure (see, e.g.,\ Figure \ref{fig:kloonimittaus}) is an acceptable (approximate) joint POVM for the initial two POVMs. Since this question is difficult to give a unified theoretical backing, we study some examples on relevant quantum observables, e.g.,\ position and momentum.

In Remark \ref{rem:measandprep}, we saw that (approximate) joint measurements for discrete PVMs can be always generated through broadcasting by a measure-and-prepare channel as long as the joint POVM is supported on the support of the product PVM (the condition of Theorem \ref{theor:PVMbroadcast}). However, usually the question is how to actually broadcast the initial state so that the subsequent local measurements suitably well approximate a joint measurement of the two PVMs, i.e.,\ the problem is that the intended joint measurement $\ms G$ is not initially known whence $\ms G$ cannot be used to define the measure-and-prepare channel. The question of how to suitably define the broadcasting channel without this crucial knowledge is a tough one in general. We next look at a couple of examples where we are able to engineer suitable broadcasting channels which realize natural approximate joint POVMs for position and momentum, first finite, then continuous.

\subsection{Approximate joint measurements for finite position and momentum through broadcasting}\label{subsec:finite}

Let us fix $D\in\N$ and denote by $\Z_D:=\Z/(D\Z)$ the $D$-element cyclic group, where the cyclic sum and subtraction will be denoted simply by $+$ and $-$, respectively. We fix an orthonormal basis $\{\fii_m\}_{m\in\Z_D}\subset\C^D$ and define the basis $\{\psi_n\}_{n\in\Z_D}\subset\C^D$ Fourier-coupled to the former one, i.e., $\psi_n=D^{-1/2}\sum_{m\in\Z_D}\<m,n\>\fii_m$ for all $n\in\Z_D$ where $\<m,n\>:=e^{i2\pi mn/D}$. Our target PVMs are now the rank-1 $\ms Q:=(|\fii_m\>\<\fii_m|)_{m\in\Z_D}$ and $\ms P:=(|\psi_n\>\<\psi_n|)_{n\in\Z_D}$; we could call these as the finite position and momentum, respectively.

We define the unitaries $U_k$ and $V_{\ell}$ for all $k,\,\ell\in\Z_D$ through $U_k\fii_m=\fii_{m+k}$ and $V_{\ell}\fii_m=\<m,\ell\>\fii_m$ for all $m\in\Z_D$. Finally, we define $W_{k,\ell}:=e^{i\pi k,\ell/D}U_kV_{\ell}$ for all $k,\,\ell\in\Z_D$. The natural approximate joint POVMs for $\ms Q$ and $\ms P$ are again of the form $\ms G^\sigma=(G^\sigma_{m,n})_{m,n\in\Z_D}$,
\begin{equation}\label{eq:covPhase}
G^\sigma_{m,n}=\frac{1}{D}W_{m,n}\sigma W_{m,n}^*,\qquad m,n\in\Z_D,
\end{equation}
where $\sigma\geq0$ is a trace-1 operator. Let us fix such an operator $\sigma$ for the duration of this example. We will now construct a broadcasting channel $\Phi:\mc T(\C^D)\to\mc T(\C^D\otimes\C^D)$ realizing $\ms G^\sigma$ through broadcasting, i.e.,\ $\Phi^*(|\fii_m\>\<\fii_m|\otimes|\psi_n\>\<\psi_n|)=G^\sigma_{m,n}$ for all $m,\,n\in\Z_D$.

It is easy to see that
\begin{equation}\label{eq:CUnitary}
U:=\sum_{k\in\Z_D}U_k\otimes|\fii_k\>\<\fii_k|=\sum_{\ell\in\Z_D}|\psi_\ell\>\<\psi_{\ell}|\otimes V_{-\ell}.
\end{equation}
This unitary can be seen as a kind of controlled transformation. Indeed, in the qubit case $D=2$, $U$ is simply the controlled-NOT unitary. Also define the (unitary) parity operator $\mc P$ through $\mc P\fii_m=\fii_{-m}$ for all $m\in\Z_D$. It is easy to check that $\mc P\psi_n=\psi_{-n}$ for all $n\in\Z_D$ as well. Let us also define the transpose $R\mapsto R^T$ w.r.t. the basis $\{\fii_m\}_{m\in\Z_D}$. We now define the broadcasting channel $\Phi$ through
$$
\Phi(\rho)=U(\mc P\sigma^T\mc P\otimes\rho)U^*,\qquad\rho\in\mc T(\C^D).
$$
This channel could be interpreted as joining the input state $\rho$ first with an ancilla in the state $\mc P\sigma^T\mc P$ and then applying the controlled transformation $U$ on the compound system. Using the two equivalent formulations of the controlled transformation $U$, we obtain
\begin{align}
\Phi^*(R\otimes S)=&\sum_{k,\ell\in\Z_D}\<\fii_k|S\fii_{\ell}\>\tr{\mc P\sigma^T\mc PU_k^*RU_{\ell}}|\fii_k\>\<\fii_{\ell}|\label{eq:muoto1}\\
=&\sum_{k,\ell\in\Z_D}\<\psi_k|R\psi_{\ell}\>\<\psi_k|\sigma\psi_{\ell}\>V_kSV_{\ell}^*\label{eq:muoto2}
\end{align}
for all $R,\,S\in\mc L(\C^D)$. Using Equations \eqref{eq:muoto1} and \eqref{eq:muoto2} respectively, we obtain
\begin{align}
\Phi^*(|\fii_m\>\<\fii_m|\otimes|\psi_n\>\<\psi_n|)=&\frac{1}{D}\sum_{k,\ell\in\Z_D}\<k-\ell,n\>\<\fii_{k-m}|\sigma\fii_{\ell-m}\>|\fii_k\>\<\fii_{\ell}|\label{eq:G1}\\
=&\frac{1}{D}\sum_{k,\ell\in\Z_D}\<m,\ell-k\>\<\psi_k|\sigma\psi_{\ell}\>|\psi_{n+k}\>\<\psi_{n+\ell}|\label{eq:G2}
\end{align}
for all $m,\,n\in\Z_D$. It is found through direct calculations that $\Phi^*(|\fii_m\>\<\fii_m|\otimes|\psi_n\>\<\psi_n|)=G^\sigma_{m,n}$ for all $m,\,n\in\Z_D$. Moreover, from Equation \eqref{eq:G1} it is easy to calculate $\sum_{n\in\Z_D}G^\sigma_{m,n}=\sum_{k\in\Z_D}\<\fii_{k-m}|\sigma\fii_{k-m}\>|\fii_k\>\<\fii_k|$ for all $m\in\Z_D$, i.e.,\ the first margin of $\ms G^\sigma$ can be seen as a convolution of the sharp $\ms Q$ with the $\ms Q$-distribution of $\mc P\sigma\mc P$. Similarly, from Equation \eqref{eq:G2}, we obtain $\sum_{m\in\Z_D}G^\sigma_{m,n}=\sum_{k\in\Z_D}\<\psi_{k-n}|\sigma\psi_{k-n}\>|\psi_k\>\<\psi_k|$ for all $n\in\Z_D$, i.e.,\ the second margin of $\ms G^\sigma$ is the convolution of $\ms P$ with the $\ms P$-distribution of $\mc P\sigma\mc P$.

Let us look at the above scheme in the case where we replace the sharp $\ms Q$ and $\ms P$ with the fuzzy $\ms Q^\lambda:=\lambda\ms Q+(1-\lambda)\ms T$ and $\ms P^\mu:=\mu\ms P+(1-\mu)\ms T$ defined as in Corollary \ref{cor:fuzzyPVM} where $\ms T=(T_m)_{m\in\Z_D}$, $T_m=(1/D)\id_D$ for all $m\in\Z_D$. For any $\sigma\in\mc T(\C^D)$, we denote $\sigma^{\ms Q}:=\sum_{m\in\Z_D}\<\fii_m|\sigma\fii_m\>|\fii_m\>\<\fii_m|$ and $\sigma^{\ms P}:=\sum_{n\in\Z_D}\<\psi_n|\sigma\psi_n\>|\psi_n\>\<\psi_n|$. It is easy to check that the margins of $\ms G^\sigma$ are given by $G^{\sigma\,(1)}_m=U_m\sigma^{\ms Q}U_m^*$ and $G^{\sigma\,(2)}_n=V_n\sigma^{\ms P}V_n^*$ for all $m,\,n\in\Z_D$. Using this, we immediately see from Corollary \ref{cor:fuzzyPVM} and Inequality \ref{eq:quditcond} therein that, for $\lambda,\,\mu\in(0,1]$ and trace-1 positive $\sigma$, $\ms G^\sigma$ can be generated through broadcasting from the POVMs $\ms Q^\lambda$ and $\ms P^\mu$ if and only if
\begin{equation}\label{eq:covkivacond}
\sigma-(1-\mu)\sigma^{\ms Q}-(1-\lambda)\sigma^{\ms P}+(1-\lambda)(1-\mu)\frac{1}{D}\id\geq0.
\end{equation}
Let us point out that, if $\lambda<1$ and $\mu<1$, a positive trace-1 $\sigma$ satisfies Inequality \eqref{eq:covkivacond} only if it is of full rank (whence $\ms G^\sigma$ is of full rank), i.e., $\ms G^\sigma$ can be generated through broadcasting from the POVMs $\ms Q^\lambda$ and $\ms P^\mu$ only if $\sigma$ is of full rank. Let us assume that $\lambda<1$ and $\mu<1$ and that trace-1 positive $\sigma$ satisfies Inequality \eqref{eq:covkivacond}. Passing Inequality \eqref{eq:covkivacond} through the completely positive map $\rho\mapsto\rho^{\ms Q}$ (and noting that $(\rho^{\ms Q})^{\ms Q}=\rho^{\ms Q}$ and $(\rho^{\ms P})^{\ms Q}=\tr{\rho}(1/D)\id$ for all $\rho\in\mc T(\C^D)$), we easily find $\sigma^{\ms Q}\geq(1-\lambda)(1/D)\id$. Similarly, passing Inequality \eqref{eq:covkivacond} through the completely positive map $\rho\mapsto\rho^{\ms P}$, we get $\sigma^{\ms P}\geq(1-\mu)(1/D)\id$. Joining these observations with Inequality \eqref{eq:covkivacond}, we find
$$
\sigma\geq(1-\mu)\sigma^{\ms Q}+(1-\lambda)\sigma^{\ms P}-(1-\lambda)(1-\mu)\frac{1}{D}\id\geq(1-\lambda)(1-\mu)\frac{1}{D}\id\neq0,
$$
i.e.,\ $\sigma$ is lower-bounded by a full-rank operator proving the claim.

We next show that, whenever the necessary and sufficient condition of Inequality \eqref{eq:covkivacond} is satisfied, $\ms G^\sigma$ can be reached from the POVMs $\ms Q^\lambda$ and $\ms P^\mu$ using a broadcasting channel of the form used earlier in this subsection. Let $\sigma$ be a trace-1 positive operator so that Inequality \eqref{eq:covkivacond} holds and define the positive trace-1 operator $\tilde{\sigma}:=(\lambda\mu)^{-1}\big(\sigma-(1-\mu)\sigma^{\ms Q}-(1-\lambda)\sigma^{\ms P}+(1-\lambda)(1-\mu)\frac{1}{D}\id\big)$. We now define the channel $\Phi:\mc T(\C^D)\to\mc T(\C^D\otimes\C^D)$ through
$$
\Phi(\rho)=U(\mc P\tilde{\sigma}^T\mc P\otimes\rho)U^*,\qquad\rho\in\mc T(\C^D).
$$
Through direct calculation, one finds that $\Phi^*(Q^\lambda_m\otimes P^\mu_n)=G^{\sigma'}_{m,n}$ for all $m,\,n\in\Z_D$ where
$$
\sigma':=\lambda\mu\tilde{\sigma}+\lambda(1-\mu)\tilde{\sigma}^{\ms Q}+(1-\lambda)\mu\tilde{\sigma}^{\ms P}+(1-\lambda)(1-\mu)\frac{1}{D}\id=\sigma
$$
where the final equality is easy to check. Thus, we actually reach $\ms G^\sigma$ in this setting with the channel $\Phi$ if and only if the minimal condition \eqref{eq:covkivacond} holds.

\subsection{Approximate joint measurements for the canonical Weyl pair through broadcasting}

We next study the continuous counterpart of the above. The system we are studying is a particle moving in one dimension or a single-mode photon field. Thus, the Hilbert space of the system is $\hil:=L^2(\R)$ in the position representation. The canonical position observable is associated with the PVM $\ms Q:\mc B(\R)\to\mc L(\hil)$ defined through $\big(\ms Q(A)\fii\big)(x)=\chi_A(x)\fii(x)$ for all $A\in\mc B(\R)$, $\fii\in\hil$, and $x\in\R$. The function $\chi_A$ here is the characteristic or indicator function of the set $A$. The canonical momentum is given by the PVM $\ms P:\mc B(\R)\to\mc L(\hil)$ which coincides with the Fourier-transformed $\ms Q$, i.e.,\ $\ms P(B)=\mc F^*\ms Q(B)\mc F$ for all $B\in\mc B(\R)$ where $\mc F\in\mc U(\hil)$ is the Fourier-Plancherel operator,
$$
(\mc F\fii)(p)=\frac{1}{\sqrt{2\pi}}\int_\R e^{-ixp}\fii(x)\,dx,\qquad\fii\in L^1(\R)\cap L^2(\R),\quad p\in\R.
$$
As is well known, $\ms Q$ and $\ms P$ are not jointly measurable, meaning that any attempt at measuring $\ms Q$ and $\ms P$ jointly yields only an approximate joint POVM. Measuring $\ms Q$ and $\ms P$ in sequence has long been studied, e.g.,\ under the framework of the Arthurs-Kelly model \cite{ArKe65}. We next describe a set of approximate joint POVMs for $\ms Q$ and $\ms P$ which is often used. In order to do this, we define the projective Weyl representation $W:\R^2\to\mc U(\hil)$, $W(q,p)=e^{-iqP+ipQ}$ for all $q,\,p\in\R$, where $Q$ is the first moment operator of $\ms Q$ (the canonical position operator) and $P$ is the first moment operator of $\ms P$ (the canonical momentum operator). For any positive trace-class operator $\sigma\in\mc T(\hil)$ of trace 1, we may define the {\it covariant phase space POVM} $\ms G_\sigma:\mc B(\R^2)\to\mc L(\hil)$ through
$$
\ms G_\sigma(C)=\frac{1}{2\pi}\int_C W(\vec{z})\sigma W(\vec{z})^*\,d\vec{z},\qquad C\in\mc B(\R^2).
$$
In fact, a POVM $\ms G:\mc B(\R^2)\to\mc L(\hil)$ is $W$-covariant, i.e.,\ $W(q,p)\ms G(C)W(q,p)^*=\ms G\big(C+(q,p)\big)$ for all $(q,p)\in\R^2$ and $C\in\mc B(\R^2)$ if and only if $\ms G=\ms G_\sigma$ for some positive $\sigma\in\mc T(\hil)$ of trace 1 \cite[Theorem 2]{KiLaYl2006}. We now fix a positive $\sigma\in\mc T(\hil)$ of trace 1 and describe a broadcasting channel which realizes $\ms G_\sigma$ through broadcasting $\ms Q$ and $\ms P$. This is possible since clearly $\ms G_\sigma\ll\ms Q\otimes\ms P$; see Theorem \ref{theor:PVMbroadcast}.

Let us define
\begin{equation}\label{eq:jatkuvaU}
U:=e^{-iP\otimes Q}=\int_\R e^{-iqP}\otimes\ms Q(dq).
\end{equation}
This is the usual interaction unitary appearing in the standard measurement model for the position (where the probe system also described by $L^2(\R)$ appears first in the tensor product). Furthermore, we denote $\mc P:=\mc F^2$, i.e.,\ $\mc P$ is the parity operator given by $(\mc P\fii)(x)=\fii(-x)$ for all $\fii\in\hil$ and $x\in\R$. Let us give $\sigma$ the spectral resolution $\sum_i s_i|\eta_i\>\<\eta_i|$ where $\{\eta_i\}_i\subset\hil$ is an orthonormal system and $t_i\geq0$ are numbers summing up to 1. We may now set up the channel $\Phi:\mc T(\hil)\to\mc T(\hil\otimes\hil)$ through
$$
\Phi(\rho)=U(\mc P\sigma'\mc P\otimes\rho)U^*,\qquad\rho\in\mc T(\hil),
$$
where $\sigma'=\sum_i s_i|\eta_i^*\>\<\eta_i^*|$ where, in turn, $\eta_i^*(x)=\overline{\eta_i(x)}$ for all $i$ and $x\in\R$. Physically this channel corresponds to the coupling of the system with an ancilla also represented by $\hil=L^2(\R)$ in the state $\mc P\sigma'\mc P$ using the standard $\ms Q$-measurement coupling $U$ (with coupling constant 1); see the discussion on the standard measurement model (first appearing in \cite{nipa}) in Section 10.4 of \cite{kirja}. To obtain a joint POVM candidate for $\ms Q$ and $\ms P$, we read $\ms Q$ of the ancilla (or probe) and carry out an exact $\ms P$-measurement on the initial system. Through straightforward calculation using the integral formula of Equation \eqref{eq:jatkuvaU}, we find that
$$
\tr{\rho\Phi^*\big(\ms Q(A)\otimes\ms P(B)\big)}=\sum_i\int_A\tr{\rho K_i(x)^*\ms P(B)K_i(x)}\,dx
$$
where $K_i(x)=\eta_i(Q-x)^*=\int_\R\overline{\eta_i(y-x)}\,d\ms Q(y)$ for all $i$ and $x\in\R$. Let us define the POVM $\ms G:\mc B(\R^2)\to\mc L(\hil)$ through $\ms G(C)=\Phi^*\big((\ms Q\otimes\ms P)(C)\big)$ for all $C\in\mc B(\R^2)$. It is easily verified that $W(q,p)K_i(x)=K_i(x+q)W(q,p)$ for all $(q,p)\in\R^2$ and $x\in\R$. Using this and the well-known (and easily verifiable) fact that $W(q,p)\ms P(B)W(q,p)^*=\ms P(B+p)$ for all $(q,p)\in\R^2$ and $B\in\mc B(\R)$, we see that $W(q,p)\ms G(A\times B)W(q,p)^*=\ms G\big((A+q)\times(B+p)\big)$ for all $(q,p)\in\R^2$ and $A,\,B\in\mc B(\R)$. From this it easily follows that $\ms G$ is $W$-covariant, implying that there is a positive $\tilde{\sigma}\in\mc T(\hil)$ of trace 1 such that $\ms G=\ms G_{\tilde{\sigma}}$.

\begin{center}
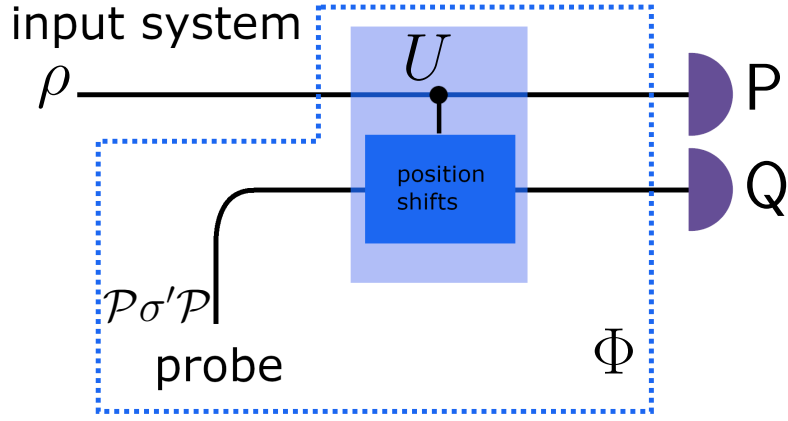
\begin{figure}
\begin{overpic}[scale=0.7,unit=1mm]{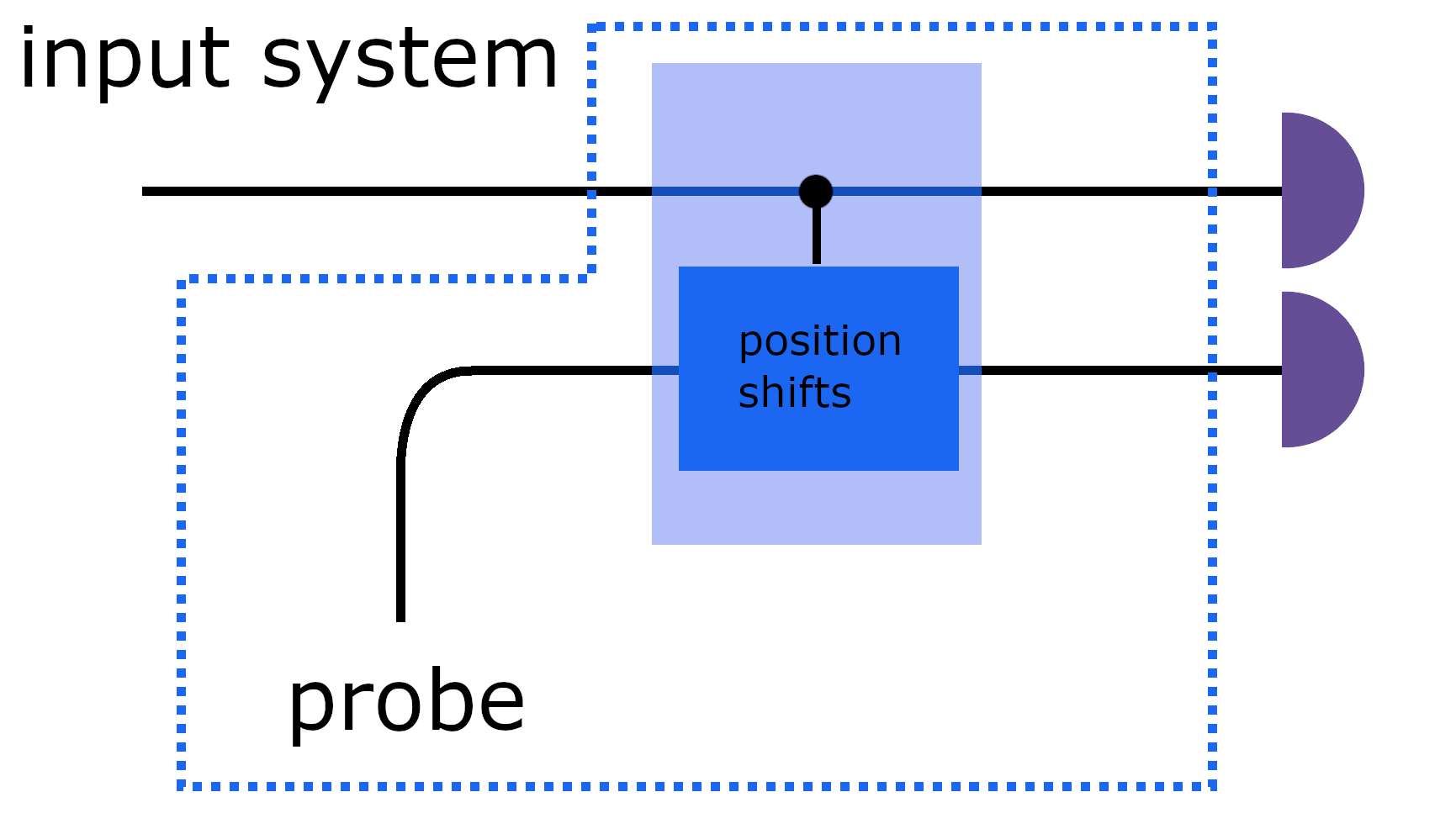}
\put(11.5,15){\begin{Large}
$\mc P\sigma'\mc P$
\end{Large}}
\put(2,45){\begin{huge}
$\rho$
\end{huge}}
\put(50,47){\begin{Huge}
$U$
\end{Huge}}
\put(75,8){\begin{Huge}
$\Phi$
\end{Huge}}
\put(95,30){\begin{Huge}
$\ms Q$
\end{Huge}}
\put(95,43){\begin{Huge}
$\ms P$
\end{Huge}}
\end{overpic}
\caption{\label{fig:coupling} Illustration of the broadcasting channel $\Phi$ (within the dashed line): Both the system and the probe are described by $L^2(\R)$ and they are jointly acted on by the unitary $U$ which couples the position of the initial system with position shifts (moves the pointer) of the probe system. This followed by reading the position of the probe (or of the pointer) corresponds to the standard position measurement. In addition to this we make a precise momentum measurement on the post-measurement system to realize the phase space observable $\ms G_\sigma$.}
\end{figure}
\end{center}

To show that $\tilde{\sigma}=\sigma$, let us directly calculate the form of $\ms G$. For this, let us make some definitions: for all $i$ and $x,\,y\in\R$, let us define $f_{i,x},\,F_{i,x,y}\in L^2(\R)$ through $f_{i,x}(z)=\sqrt{s_i}\overline{\eta_i(z-x)}$ and $F_{i,x,y}(z)=\overline{(\mc F f_{i,x})(y-z)}$. For $f\in L^2(\R)$, we also denote $\hat{f}:=\mc F f$ for brevity. Using the fact that $\mc F(\fii\cdot\psi)=(2\pi)^{-1/2}\hat{\fii}*\hat{\psi}$, where $\fii\cdot\psi$ denotes the pointwise product of the functions $\fii$ and $\psi$ and $\hat{\fii}*\hat{\psi}$ is the convolution of the Fourier transforms of $\fii$ and $\psi$, we have, for all $A,\,B\in\mc B(\R)$ and $\fii\in\hil$,
\begin{align*}
\<\fii|\ms G(A\times B)\fii\>=&\sum_i\int_A\<K_i(q)\fii|\ms P(B)K_i(q)\fii\>\,dq=\sum_i\int_A\<f_{i,q}\cdot\fii|\ms P(B)(f_{i,q}\cdot\fii)\>\,dq\\
=&\sum_i\int_A\<\mc F(f_{i,q}\cdot\fii)|\ms Q(B)\mc F(f_{i,q}\cdot\fii)\>\,dq=\frac{1}{2\pi}\sum_i\int_A\<\hat{f}_{i,q}*\hat{\fii}|\ms Q(B)(\hat{f}_{i,q}*\hat{\fii})\>\,dq\\
=&\frac{1}{2\pi}\sum_i\int_A\int_B|(\hat{f}_{i,q}*\hat{\fii})(p)|^2\,dp\,dq\\
=&\frac{1}{2\pi}\sum_i\int_A\int_B\int_\R\int_\R\overline{\hat{f}_{i,q}(p-z)\hat{\fii}(z)}\hat{f}_{i,q}(p-w)\hat{\fii}(w)\,dz\,dw\,dp\,dq\\
=&\frac{1}{2\pi}\sum_i\int_A\int_B\<\hat{\fii}|F_{i,q,p}\>\<F_{i,q,p}|\hat{\fii}\>\,dp\,dq\\
=&\frac{1}{2\pi}\sum_i\int_A\int_B\<\fii|\mc F^*F_{i,q,p}\>\<\mc F^*F_{i,q,p}|\fii\>\,dp\,dq.
\end{align*}
Using the fact that $(\mc F^*f^*)(z)=\overline{(\mc F^*f)(-z)}$ where $f^*$ is defined through $f^*(z)=\overline{f(z)}$ and that $\mc F\mc P=\mc P\mc F$, we have
\begin{align*}
(\mc F^*F_{i,q,p})(z)=&\overline{(\mc F^*e^{-ipP}\mc P\mc Ff_{i,q})(-z)}=\overline{(e^{ipQ}\mc Pf_{i,q})(-z)}\\
=&e^{ipz}\overline{f_{i,x}(z)}=\sqrt{s_i}e^{ipz}\eta_i(z-q)=\sqrt{s_i}e^{iqp/2}\big(W(q,p)\eta_i)(z)
\end{align*}
for all $i$ and $q,\,p,\,z\in\R$ where the final form can be easily derived through a direct calculation using Baker-Campbell-Hausdorff identities. Combining this with the preceding calculation, we finally have, for all $A,\,B\in\mc B(\R)$ and $\fii\in\hil$,
$$
\<\fii|\ms G(A\times B)\fii\>=\frac{1}{2\pi}\sum_i s_i\int_A\int_B\<\fii|W(q,p)\eta_i\>\<W(q,p)\eta_i|\fii\>\,dp\,dq=\<\fii|\ms G_\sigma(A\times B)\fii\>,
$$
showing that, indeed, $\ms G=\ms G_\sigma$, which we originally set out to prove.

\section{Symmetric broadcasting}\label{sec:symm}

We next look into the case were the two POVMs we are broadcasting and the target POVM are all symmetric. We describe symmetry here through covariance under projective unitary representations of a group. We concentrate on compact groups because we get simple conditions on broadcastability in this setting. In fact, our main results deal with symmetry w.r.t. finite groups. Let us first recall some terminology: Let $\mc G$ be a group with the neutral element $e$ and $\hil$ be a Hilbert space. A map $U:\mc G\to\mc U(\hil)$ is a {\it projective unitary representation} if $U(e)=\id_\hil$ and there is a map $m:\mc G\times\mc G\to\T$ (where $\T$ is the set of complex numbers of modulus 1) which is a {\it multiplier}, i.e.,\ $m(e,g)=1=m(g,e)$ for all $g\in\mc G$ and
$$
m(g,h)m(gh,k)=m(g,hk)m(h,k),\qquad g,\,h,\,k\in\mc G,
$$
such that $U(gh)=m(g,h)U(g)U(h)$ for all $g,\,h\in\mc G$. We typically concentrate on locally compact groups (or measurable groups more generally) where we tacitly assume that the multiplier $m$ is also measurable. It is well-known that there is a many-to-one correspondence between projective unitary representations and group homomorphisms $\mc G\ni g\mapsto\beta_g\in{\rm Aut}\big(\mc L(\hil)\big)$ of $\mc G$ into the group of automorphisms of $\mc L(\hil)$ where the homomorphism or {\it action} corresponding to a projective unitary representation $U$ is given by $\beta_g(R)=U(g)RU(g)^*$ for all $g\in\mc G$ and $R\in\mc L(\hil)$; see Sections VII.1 and VII.2 of \cite{Varadarajan}. Thus, a projective unitary representation $U$ is a general description of dynamics or symmetries of states where the state transforms under a group element $g\in\mc G$ as $\rho\mapsto U(g)\rho U(g)^*$.

Let now $\mc G$ be a compact group with the normalized left Haar measure $\mu:\mc B(\mc G)\to[0,1]$, $\mu(\mc G)=1$. Let $(X,\mc A)$ and $(Y,\mc B)$ be measurable spaces modelling the value spaces of the POVMs to be broadcast. We assume that these are $\mc G$-spaces. This means that, e.g.,\ for $(X,\mc A)$, there is a $\big(\mc B(\mc G)\otimes\mc A,\mc A\big)$-measurable map $\mc G\times X\ni(g,x)\mapsto gx\in X$ such that $ex=x$ for all $x\in X$ and $g(hx)=(gh)x=:ghx$ for all $g,\,h\in\mc G$ and $x\in X$. For a projective unitary representation $U:\mc G\to\mc U(\hil)$, we say that a POVM $\ms M:\mc A\to\mc L(\hil)$ is {\it $(U,X,\mc A)$-covariant} if
$$
U(g)\ms M(A)U(g)^*=\ms M(gA)
$$
for all $g\in\mc G$ and $A\in\mc A$. From now on, we fix projective unitary representations $U,\,V:\mc G\to\mc U(\hil)$, a $(U,X,\mc A)$-covariant POVM $\ms M$, and a $(V,Y,\mc B)$-covariant POVM $\ms N$. Let us fix another projective unitary representation $W:\mc G\to\mc U(\hil)$. We view $(X\times Y,\mc A\otimes\mc B)$ as a $\mc G$-space where $g(x,y):=(gx,gy)$ for all $g\in\mc G$ and $(x,y)\in X\times Y$. Our target POVM is now $\ms G$ which we assume to be $(W,X\times Y,\mc A\otimes\mc B)$-covariant.

To describe the relevant broadcasting schemes, let us briefly describe covariance in channels. For projective unitary representations $W_1:\mc G\to\mc U(\hil)$ and $W_2:\mc G\to\mc U(\mc K)$, we say that a channel $\Phi:\mc T(\hil)\to\mc T(\mc K)$ is {\it $(W_1,W_2)$-covariant} if
$$
\Phi\big(W_1(g)\rho W_1(g)^*\big)=W_2(g)\Phi(\rho)W_2(g)^*
$$
for all $g\in\mc G$ and $\rho\in\mc T(\hil)$. We next show that if the above $(W,X\times Y,\mc A\otimes\mc B)$-covariant POVM $\ms G$ can be generated from the $(U,X,\mc A)$-covariant $\ms M$ and the $(V,Y,\mc B)$-covariant $\ms N$ through broadcasting, then the broadcasting channel can be assumed to be $(W,U\otimes V)$-covariant where $U\otimes V:\mc G\to\mc U(\hil\otimes\hil)$, $(U\otimes V)(g)=U(g)\otimes V(g)$ for all $g\in\mc G$. Indeed, if $\Psi^*(\ms M\otimes\ms N)=\ms G$ for some channel $\Psi$, we can define another channel $\Phi:\mc T(\hil)\to\mc T(\hil\otimes\hil)$ through
$$
\tr{\Phi(\rho)S}=\int_{\mc G}\tr{\Psi\big(W(g)^*\rho W(g)\big)\big(U(g)\otimes V(g)\big)^*S\big(U(g)\otimes V(g)\big)}\,d\mu(g)
$$
for all $\rho\in\mc T(\hil)$ and $S\in\mc L(\hil\otimes\hil)$. Let us first show that $\Phi$ is $(W,U\otimes V)$-covariant. Let $\rho\in\mc T(\hil)$, $S\in\mc L(\hil\otimes\hil)$, and $g\in\mc G$. We have
\begin{align*}
&\tr{\Phi\big(W(g)\rho W(g)^*\big)S}\\
=&\int_{\mc G}\tr{\Psi\big(W(g')^*W(g)\rho W(g)^*W(g')\big)\big(U(g')\otimes V(g')\big)^*S\big(U(g')\otimes V(g')\big)}\,d\mu(g')\\
=&\int_{\mc G}\tr{\Psi\big(W(g^{-1}g')^*\rho W(g^{-1}g')\big)\big(U(g')\otimes V(g')\big)^*S\big(U(g')\otimes V(g')\big)}\,d\mu(g')\\
=&\int_{\mc G}\tr{\Psi\big(W(g')^*\rho W(g')\big)\big(U(gg')\otimes V(gg')\big)^*S\big(U(gg')\otimes V(gg')\big)}\,d\mu(g')\\
=&\tr{\Phi(\rho)\big(U(g)\otimes V(g)\big)^*S\big(U(g)\otimes V(g)\big)},
\end{align*}
showing that $\Phi$ is $(W,U\otimes V)$-covariant. Note that we have used the properties of the left Haar measure in the third equality above. Let us next show that $\Phi^*(\ms M\otimes\ms N)=\ms G$. Let $\rho\in\mc T(\hil)$, $A\in\mc A$, and $B\in\mc B$. We have
\begin{align*}
\tr{\rho\Phi^*\big(\ms M(A)\otimes\ms N(B)\big)}=&\int_{\mc G}\tr{\Psi\big(W(g)^*\rho W(g)\big)\big(U(g)^*\ms M(A)U(g)\otimes V(g)^*\ms N(B)V(g)\big)}\,d\mu(g)\\
=&\int_{\mc G}\tr{\rho W(g)\Psi^*\big(\ms M(g^{-1}A)\otimes\ms N(g^{-1}B)\big)W(g)^*}\,d\mu(g)\\
=&\int_{\mc G}\tr{\rho W(g)\ms G(g^{-1}A\times g^{-1}B)W(g)^*}\,d\mu(g)\\
=&\int_{\mc G}\tr{\rho\ms G(A\times B)}\,d\mu(g)=\mu(\mc G)\tr{\rho\ms G(A\times B)}=\tr{\rho\ms G(A\times B)},
\end{align*}
proving the claim.

In the typical situation, $(X,\mc A)$ and $(Y,\mc B)$ are transitive w.r.t. $\mc G$. This means that, e.g.,\ for $(X,\mc A)$, for all $x,x'\in X$ there is $g\in\mc G$ such that $gx=x'$. Fixing a point $x_0\in X$ and the stability subgroup $H:=\{h\in\mc G\,|\,hx_0=x_0\}$, it is now easy to see that, as sets, $X=\mc G/H$ where the coset space $\mc G/H$ is a $\mc G$-space in the sense that $g(g'H)=gg'H$ for all $g,\,g'\in\mc G$. As long as $X$ is a locally compact and second-countable topological space and $\mc A=\mc B(X)$ is standard Borel, this identification is actually a homeomorphism \cite[Theorem 5.11]{Varadarajan}. This is why we assume from now on that there are closed subgroups $H,\,K\leq\mc G$ such that $X=\mc G/H$, $Y=\mc G/K$, $\mc A=\mc B(\mc G/H)$, and $\mc B=\mc B(\mc G/K)$. In this setting, we call $(U,X,\mc A)$-covariance simply as $(U,\mc G/H)$-covariance and similarly for $(V,Y,\mc B)$. Moreover, we assume that $\mc G$ is finite and study the discrete (finite) POVMs $\ms M=(M_x)_{x\in\mc G/H}$ and $\ms N=(N_y)_{y\in\mc G/K}$ where, e.g.,\ $M_x:=\ms M(\{x\})$ for all $x\in\mc G/H$. It is easy to see that, due to covariance, $\ms M$ and $\ms N$ are completely characterized by the operators $M_H$ and $N_H$ through $M_{gH}=U_gM_H U_g^*$ and $N_{gK}=V_gN_H V_g^*$ for all $g\in\mc G$. Note that we are using lower indices for the arguments of representations, multipliers, etc.\ in this finite case.

Our final simplification allows us to write very simple necessary and sufficient conditions for broadcastability. We assume that $\mc G=\mc X\times\mc Y$ where $\mc X$ and $\mc Y$ are finite groups and $H=\{e_1\}\times\mc Y$ and $K=\mc X\times\{e_2\}$ where $e_1\in\mc X$ and $e_2\in\mc Y$ are the neutral elements. Thus, $\mc G/H\simeq\mc X$ and $\mc G/K\simeq\mc Y$ and $\mc G/H\times\mc G/K\simeq\mc G$. The latter fact means that the value set of the target POVMs is identifiable with the entire group $\mc G$ which is trivially a transitive $\mc G$-space. To sum up, for $(x,y)\in\mc G=\mc X\times\mc Y$, $x'\in\mc X$, and $y'\in\mc Y$, $(x,y)x'=xx'$ and $(x,y)y'=yy'$. Finally, a $(W,\mc G)$-covariant POVM $\ms G=(G_g)_{g\in\mc G}$ can be generated through broadcasting from $\ms M$ and $\ms N$ if and only if there is a $(W,U\otimes V)$-covariant channel $\Phi$ such that $\Phi^*(M_H\otimes N_K)=G_e$. To see this, let $\Phi:\mc T(\hil)\to\mc T(\hil\otimes\hil)$ be a channel such that $\Phi^*(\ms M\otimes\ms N)=\ms G$. As we have already shown, we can assume that $\Phi$ is $(W,U\otimes V)$-covariant. It immediately follows that $\Phi^*(M_H\otimes N_K)=G_{H\times K}=G_e$ using the identifications established earlier. On the other hand, if $\Phi^*(M_H\otimes N_K)=G_e$ for a $(W,U\otimes V)$-covariant channel $\Phi$, we have, for all $x\in\mc X\simeq\mc G/H$ and $y\in\mc Y\simeq\mc G/K$,
\begin{align*}
\Phi^*(M_x\otimes N_y)=&\Phi^*(U_{x,y}M_HU_{x,y}^*\otimes V_{x,y}N_KV_{x,y}^*)=W_{x,y}\Phi^*(M_H\otimes N_K)W_{x,y}^*\\
=&W_{x,y}G_{e_1,e_2}W_{x,y}^*=G_{x,y},
\end{align*}
showing that $\ms G$ can indeed be generated through broadcasting from $\ms M$ and $\ms N$. We will utilize this result in the case of finite covariant phase space POVMs where we obtain particularly simple conditions from these observations.

\subsection{Covariant broadcasting of generalized finite position and momentum}

We use the notations of Subsection \ref{subsec:finite} and consider the situation like that above where $\hil=\C^D$, $\mc G=\Z_D^2$, $\mc X=\Z_D=\mc Y$, and $U=V=W$ is the Weyl representation. Denoting $\omega(k,\ell;m,n):=(2\pi/D)(kn-\ell m)$, the multiplier of $W$ is of the form $m_{k,\ell;m,n}=e^{\frac{i}{2}\omega(k,\ell;m,n)}$ and we have the canonical commutation relations (CCR)
\begin{align}\label{eq:CCR}
W_{k,\ell}W_{m,n}=e^{-i\omega(k,\ell;m,n)}W_{m,n}W_{k,\ell},\qquad W_{-k,-\ell}=W_{k,\ell}^*
\end{align}
for all $(k,\ell),\,(m,n)\in\Z_D^2$.

According to Theorem 4.6 of \cite{Haapasalo2019}, a channel $\Phi:\mc T(\hil)\to\mc T(\hil\otimes\hil)$ is $(W,W\otimes W)$-covariant if and only if there are $f_{k,\ell;m,n}\in\C$ for $(k,\ell),\,(m,n)\in\Z_D^2$ such that $f_{0,0;0,0}=1$, the map
$$
\Z_D^4\times\Z_D^4\ni(k,\ell,m,n;k',\ell',m',n')\mapsto e^{-\frac{i}{2}\big(\omega(k,\ell;m',n')+\omega(m,n;k',\ell')\big)}f_{k'-k,\ell'-\ell;m'-m,n'-n}\in\C
$$
is a positive kernel (i.e.,\ the multi-index matrix whose entries are those on the right-hand side above is positive semi-definite), and $\Phi^*(W_{k,\ell}\otimes W_{m,n})=f_{k,\ell;m,n}W_{k+m,\ell+n}$ for all $(k,\ell),\,(m,n)\in\Z_D^2$. Note that there is a unique $\tau\in\mc S(\C^D)$ such that $f_{0,\ell;m,0}=\tr{W_{\ell,m}^*\tau}$ for all $\ell,\,m\in\Z_D$, i.e.,\ $(\ell,m)\mapsto f_{0,\ell;m,0}$ is the characteristic function of $\tau$; in the sequel, we call this as the state corresponding to the covariant channel $\Phi$. On the other hand, for any $\tau\in\mc S(\C^D)$, there is a $(W,W\otimes W)$-covariant channel $\Phi$ such that the corresponding numbers $f_{k,\ell;m,n}$ are such that $f_{0,\ell;m,0}=\tr{W_{\ell,m}^*\tau}$ for all $\ell,\,m\in\Z_D$ and, otherwise, $f_{k,\ell;m,n=0}$. Indeed, it is straight-forward to check that the channel $\Phi$ defined through
\begin{equation}\label{eq:tauChan}
\Phi(\rho)=\sum_{r,s\in\Z_D}\tr{\rho G^\tau_{r,s}}\,|\fii_r\>\<\fii_r|\otimes|\psi_s\>\<\psi_s|,\qquad\rho\in\mc T(\C^D),
\end{equation}
satisfies this condition. Recall the definition of the POVM $\ms G^\tau=(G^\tau_{r,s})_{(r,s)\in\Z_D^2}$ in Equation \eqref{eq:covPhase} and the position basis $\{\fii_k\}_{k\in\Z_D}$ and the momentum basis $\{\psi_\ell\}_{\ell\in\Z_D}$ of Subsection \ref{subsec:finite}.

We are interested in broadcasting local generalized position $\ms M=(M_k)_{k\in\Z_D}$ and momentum $\ms N=(N_\ell)_{\ell\in\Z_D}$. This means that $W_{m,n}M_k W_{m,n}^*=M_{k+m}$ and $W_{m,n}N_\ell W_{m,n}^*=N_{\ell+n}$ for all $k,\,\ell,\,m,\,n\in\Z_D$. As $\ms M$ thus commutes with $V_\ell$ for all $\ell\in\Z_D$, all $M_k$ diagonalize in the position basis. Similarly, all $N_\ell$ diagonalize in the momentum basis. Especially, $M_0=\sum_{k\in\Z_D}\mu_k|\fii_k\>\<\fii_k|$ and $N_0=\sum_{\ell\in\Z_D}\nu_\ell|\psi_\ell\>\<\psi_\ell|$ for some probability distributions $(\mu_k)_{k\in\Z_D}$ and $(\nu_\ell)_{\ell\in\Z_D}$. Our target POVMs are $(W,\Z_D^2)$-covariant, i.e.,\ POVMs of the form $\ms G=(G_{k,\ell})_{(k,\ell)\in\Z_D^2}$ such that $W_{m,n}G_{k,\ell}W_{m,n}^*=G_{k+m,\ell+n}$ for all $k,\,\ell,\,m,\,n\in\Z_D$. These are all exactly of the form $\ms G^\sigma$ for some $\sigma\in\mc S(\C^D)$. We are now able to write down simple necessary and sufficient conditions for the broadcastability of these POVMs into a $(W,\Z_D^2)$-covariant POVM.

\begin{proposition}
Let $\ms M=(M_k)_{k\in\Z_D}$ be a generalized position and $\ms N=(N_\ell)_{\ell\in\Z_D}$ be a generalized momentum POVM defined by probability distributions $(\mu_k)_{k\in\Z_D}$ and $(\nu_\ell)_{\ell\in\Z_D}$ through $M_0=\sum_{k\in\Z_D}\mu_k|\fii_k\>\<\fii_k|$ and $N_0=\sum_{\ell\in\Z_D}\nu_\ell|\psi_\ell\>\<\psi_\ell|$. The POVM $\ms G^\sigma$ can be generated through broadcasting from $\ms M$ and $\ms N$ if and only if there is a state $\tau$ such that
\begin{equation}\label{eq:sigmatau}
\sigma=\sum_{k,\ell\in\Z_D}\mu_k\nu_\ell\,W_{k,\ell}\tau W_{k,\ell}^*.
\end{equation}
\end{proposition}

\begin{proof}
Let us assume that there is a channel $\Phi$ such that $\Phi^*(\ms M\otimes\ms N)=\ms G^\sigma$. As we have seen, we may assume that $\Phi$ is $(W,W\otimes W)$-covariant. Let $\tau_0$ be the associated state, i.e.,
$$
\Phi^*(W_{0,\ell}\otimes W_{m,0})=\tr{W_{\ell,m}^*\tau_0}W_{m,\ell},\qquad\ell,\,m\in\Z_D.
$$
According to our earlier discussion, $\Phi^*(\ms M\otimes\ms N)=\ms G^\sigma$ is equivalent with $\Phi^*(M_0\otimes N_0)=G^\sigma_{0,0}=D^{-1}\sigma$. Using the easily verifiable fact that $\{D^{-1/2}W_{k,\ell}\}_{(k,\ell)\in\Z_D^2}$ is an orthonormal basis of the Hilbert-Schmidt space and that $\tr{W_{k,\ell}^*M_0}=\delta_{k,0}\sum_s\mu_s\overline{\<\ell,s\>}$ and $\tr{W_{m,n}^*N_0}=\delta_{n,0}\sum_r\nu_r\<m,r\>$, this means
\begin{align*}
\frac{1}{D}\sigma=&\Phi^*(M_0\otimes N_0)=\frac{1}{D^2}\sum_{k,\ell,m,n\in\Z_D}\tr{W_{k,\ell}^*M_0}\tr{W_{m,n}^*N_0}\Phi^*(W_{k,\ell}\otimes W_{m,n})\\
=&\frac{1}{D^2}\sum_{\ell,m,r,s\in\Z_D}\mu_s\nu_r\underbrace{\overline{\<\ell,s\>}\<m,r\>}_{=e^{-i\omega(\ell,m;r,s)}}\Phi^*(W_{0,\ell}\otimes W_{m,0})\\
=&\frac{1}{D^2}\sum_{\ell,m,r,s\in\Z_D}\mu_s\nu_r\tr{W_{\ell,m}^*\tau_0}e^{-i\omega(\ell,m;r,s)}W_{m,\ell}.
\end{align*}
Defining the product distribution $\mu\times\nu=(\mu_k\nu_\ell)_{(k,\ell)\in\Z_D^2}$ and its Fourier-transform $\widehat{\mu\times\nu}$ through
$$
(\widehat{\mu\times\nu})_{\ell,m}=\frac{1}{D}\sum_{r,s\in\Z_D}e^{i\omega(\ell,m;r,s)}\mu_r\nu_s=\frac{1}{D}\sum_{r,s\in\Z_D}e^{-i\omega(m,\ell;r,s)}\mu_s\nu_r,
$$
we now see that
$$
\sigma=\sum_{\ell,m\in\Z_D}(\widehat{\mu\times\nu})_{m,\ell}\tr{W_{\ell,m}^*\tau_0}W_{m,\ell}.
$$

Let us denote by $R\mapsto R^T$ the transpose w.r.t.\ the position basis of $\C^D$. It is easy to check that $\mc F^*W_{k,\ell}^T\mc F=W_{\ell,k}$. Moreover, $\mc F^T=\mc F$ showing that the map $R\mapsto\mc F^*R^T\mc F$ is an involution. Using this, we have
\begin{equation}\label{eq:pseudotranspose}
\mc F^*\sigma^T\mc F=\sum_{\ell,m\in\Z_D}(\widehat{\mu\times\nu})_{m,\ell}\tr{W_{\ell,m}^*\tau_0}W_{\ell,m}.
\end{equation}
For any $\alpha_{r,s}\in\C$ ($r,\,s\in\Z_D$) and $\tau\in\mc S(\C^D)$, it is easy to check that, for $\tau':=\sum_{r,s}\alpha_{r,s}W_{r,s}\tau W_{r,s}$, we have $\tr{W_{\ell,m}^*\tau'}=D\hat{\alpha}_{\ell,m}\tr{W_{\ell,m}^*\tau}$ for all $\ell,\,m\in\Z_D$ where the Fourier-transform $\hat{\alpha}$ is defined in the same way as that for the product distribution $\mu\times\nu$ above. This follows easily from the CCR conditions \eqref{eq:CCR}. Comparing this observation to Equation \eqref{eq:pseudotranspose}, we have $\hat{\alpha}_{\ell,m}=(\widehat{\mu\times\nu})_{m,\ell}$ for all $\ell,\,m\in\Z_D$. Using the (easily verifiable) inversion formula
$$
\alpha_{r,s}=\frac{1}{D}\sum_{\ell,m\in\Z_D}e^{-i\omega(\ell,m;r,s)}\hat{\alpha}_{\ell,m},\qquad r,\,s\in\Z_D,
$$
we have, for all $r,\,s\in\Z_D$,
\begin{align*}
\alpha_{r,s}=&\frac{1}{D}\sum_{\ell,m\in\Z_D}e^{-i\omega(\ell,m;r,s)}(\widehat{\mu\times\nu})_{m,\ell}=\frac{1}{D^2}\sum_{j,k,\ell,m\in\Z_D}e^{-i\omega(\ell,m;r,s)}e^{i\omega(m,\ell;j,k)}\mu_j\nu_k\\
=&\frac{1}{D^2}\sum_{j,k,\ell,m\in\Z_D}e^{-i\omega(\ell,m;r+k,s+j)}\mu_j\nu_k=\mu_{-s}\nu_{-r}.
\end{align*}
Thus, $\mc F^*\sigma^T\mc F=\sum_{r,s}\mu_{-s}\nu_{-r}W_{r,s}\tau_0 W_{r,s}^*$, implying
\begin{align*}
\sigma=&\sum_{r,s\in\Z_D}\mu_{-s}\nu_{-r}\big(\mc FW_{r,s}\tau_0W_{r,s}^*\mc F^*\big)^T=\sum_{r,s\in\Z_D}\mu_{-s}\nu_{-r}\mc F^{*\,T}W_{r,s}^{*\,T}\tau_0^TW_{r,s}^T\mc F^T\\
=&\sum_{r,s\in\Z_D}\mu_{-s}\nu_{-r}\underbrace{\mc F^*W_{r,s}^{*\,T}\mc F}_{=W_{s,r}^*}\mc F^*\tau_0^T\mc F\underbrace{\mc F^*W_{r,s}^T\mc F}_{=W_{s,r}}=\sum_{r,s\in\Z_D}\mu_{-s}\nu_{-r}W_{s,r}^*\mc F^*\tau_0^T\mc FW_{s,r}\\
=&\sum_{r,s\in\Z_D}\mu_{-r}\nu_{-s}W_{r,s}^*\mc F^*\tau_0^T\mc FW_{r,s}=\sum_{r,s\in\Z_D}\mu_r\nu_sW_{r,s}\mc F^*\tau_0^T\mc FW_{r,s}^*.
\end{align*}
Defining $\tau:=\mc F^*\tau_0^T\mc F$, we obtain Equation \eqref{eq:sigmatau}.

Assume now that Equation \eqref{eq:sigmatau} holds with some state $\tau$. Using $\tau$, define the channel $\Phi$ as in Equation \eqref{eq:tauChan}. We now have
\begin{align*}
\Phi^*(M_0\otimes N_0)=&\sum_{r,s\in\Z_D}\<\fii_r|M_0\fii_r\>\<\psi_s|N_0\psi_s\>\,G^\tau_{r,s}=\sum_{r,s\in\Z_D}\mu_r\nu_s\,G^\tau_{r,s}\\
=&\frac{1}{D}\sum_{r,s\in\Z_D}\mu_r\nu_s\,W_{r,s}\tau W_{r,s}^*=\frac{1}{D}\sigma=G^\sigma_{0,0},
\end{align*}
and since $\Phi$ is $(W,W\otimes W)$-covariant, we have $\Phi^*(\ms M\otimes\ms N)=\ms G^\sigma$.
\end{proof}

\section{POVMs reached through broadcasting, local measurements, and post-processing}\label{sec:BLMPP}

We now turn our attention to a scenario where we, as earlier, first broadcast the input state with a broadcasting channel $\Phi$, then carry out {\it fixed} local measurements $\ms M_0$ and $\ms N_0$ on the `clones', and finally classically manipulate, i.e.,\ post-process, the product observable $\ms M_0\otimes\ms N_0$. Let us formalize the concept of post-processing: For any general POVM $\ms M:\mc A\to\mc L(\hil)$ (with value space $(X,\mc A)$ and Hilbert space $\hil$) any target value space $(Z,\mc C)$, we say that $\beta:\mc C\times X\to\R$ is a {\it $\ms M$-weak Markov kernel} if
\begin{itemize}
\item for all $C\in\mc C$, $0\leq\beta(C|x)\leq1$ for $\ms M$-a.a.\ $x\in X$,
\item $\beta(\emptyset|x)=0$ and $\beta(Z|x)=1$ for $\ms M$-a.a.\ $x\in X$, and
\item for any disjoint sequence $C_1,\,C_2,\ldots\in\mc C$, $\beta\big(\cup_{i=1}^\infty C_i\big|x\big)=\sum_{i=1}^\infty\beta(C_i|x)$ for $\ms M$-a.a.\ $x\in X$.
\end{itemize}
Moreover, we define the post-processed POVM $\beta[\ms M]:\mc C\to\mc L(\hil)$ through
$$
\beta[\ms M](C)=\int_X\beta(C|x)\,d\ms M(x),\qquad C\in\mc C.
$$

Let us assume that the value space of $\ms M_0$ is $(X,\mc A)$ and that of $\ms N_0$ is $(Y,\mc B)$. Moreover, let the Hilbert space of $\ms M_0$ be $\mc K_1$ and that of $\ms N_0$ be $\mc K_2$. Let us pick a target value space (which will be the value space of the target POVM generated in this setting) $(Z,\mc C)$; typically $Z=X'\times Y'$ and $\mc C=\mc A'\otimes\mc B'$ where $(X',\mc A')$ and $(Y',\mc B')$ are some measurable spaces, but we do not have to make this specification for now. Post-processing is now represented by a $\ms M_0\otimes\ms N_0$-weak Markov kernel $\beta(\cdot|\cdot):\mc C\times(X\times Y)\to\R$, so that the POVM generated by the above scheme is $\ms G=\beta[\Phi^*\circ(\ms M_0\otimes\ms N_0)]$, i.e., for all $C\in\mc C$,
$$
\ms G(C)=\int_{X\times Y}\beta(C|x,y)\Phi^*\big(\ms M_0(dx)\otimes\ms N_0(dy)\big).
$$
We call the above scheme of generating POVMs $\ms G$ as the {\it broadcasting, local measurements, and post-processing (BLMPP) protocol} which is illustrated in Figure \ref{fig:blmpp}. This protocol uses the fixed pair $(\ms M_0,\ms N_0)$ as its resource from which different (joint) measurements is derived using varying the broadcasting channel $\Phi$, the target value space $(Z,\mc C)$, and the post-processing $\beta$.

\begin{center}
\begin{figure}
\begin{overpic}[scale=0.35,unit=1mm]{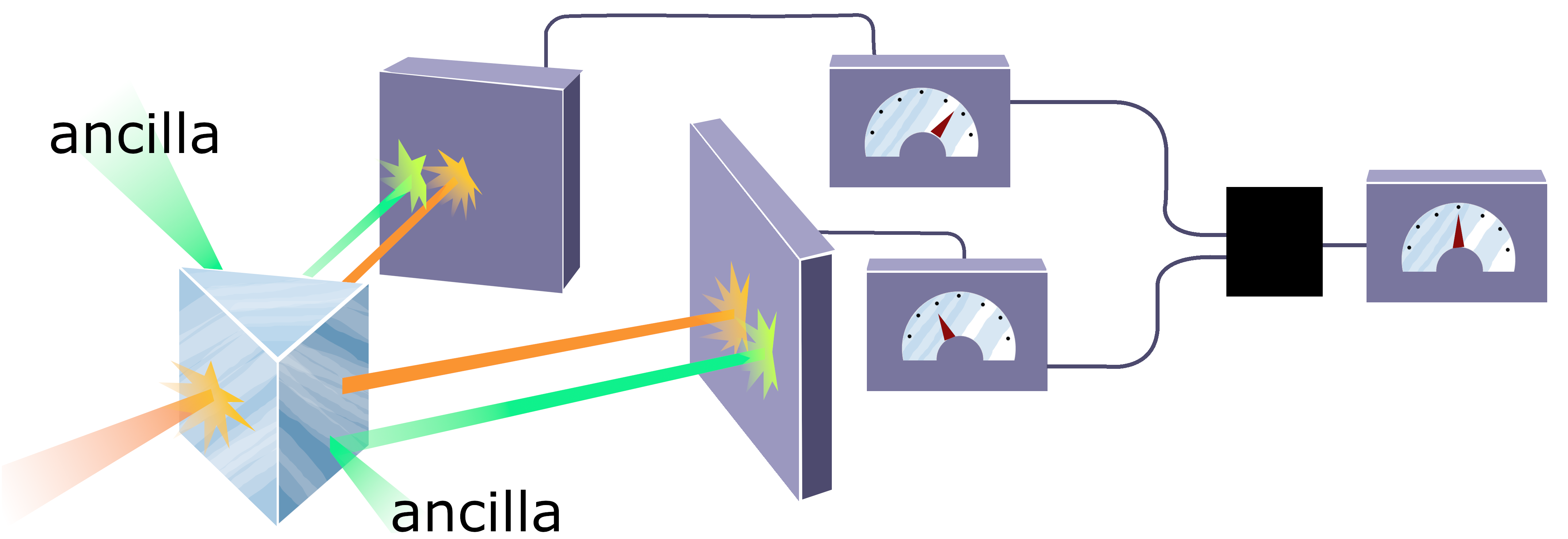}
\put(0,3){\begin{Huge}
$\rho$
\end{Huge}}
\put(13,15){\begin{Huge}
$\Phi$
\end{Huge}}
\put(81,38){\begin{Large}
$\ms M_0$
\end{Large}}
\put(84,9){\begin{Large}
$\ms N_0$
\end{Large}}
\put(100,22){\begin{Large}
\textcolor{white}{$\beta$}
\end{Large}}
\put(113,31){\begin{Huge}
$\ms G$
\end{Huge}}
\end{overpic}
\caption{\label{fig:blmpp} An illustration of the BLMPP protocol: The input state $\rho$ is broadcast with the channel $\Phi$ (here, again, the beam-splitting crystal). The resource of this protocol is the pair $(\ms M_0,\ms N_0)$ of POVMs and to match dimensions, the dimensions of the arms of the broadcasting channel $\Phi$ typically differ from the input dimension; this is highlighted by the presence of ancillary systems. After the local measurements of $\ms M_0$ and $\ms N_0$, the classical post-processing, represented by a (weak) Markov kernel $\beta$ is carried out, realizing the POVM $\ms G$.}
\end{figure}
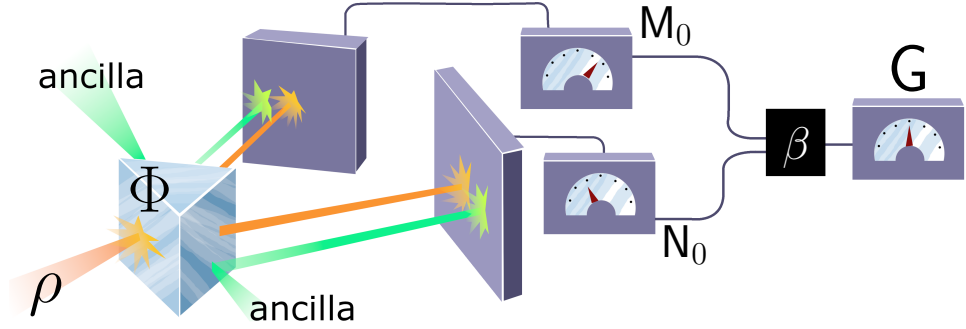
\end{center}

\begin{definition}\label{def:BLMPP}
Fix Hilbert spaces $\mc K_1$ and $\mc K_2$, value spaces $(X,\mc A)$ and $(Y,\mc B)$, and POVMs $\ms M_0:\mc A\to\mc L(\mc K_1)$ and $\ms N_0:\mc B\to\mc L(\mc K_2)$. For any Hilbert space $\hil$ and any measurable space $(Z,\mc C)$, we denote by $\mf G_{\hil,Z,\mc C}(\ms M_0,\ms N_0)$ the set of POVMs $\beta[\Phi^*\circ(\ms M_0\otimes\ms N_0)]$ where $\Phi:\mc T(\hil)\to\mc T(\mc K_1\otimes\mc K_2)$ is any channel and $\beta(\cdot|\cdot):\mc C\times(X\times Y)\to\R$ is any $\ms M_0\otimes\ms N_0$-weak Markov kernel. Moreover, we denote $\mf G_\hil(\ms M_0,\ms N_0):=\mf G_{\hil,\R,\mc B(\R)}(\ms M_0,\ms N_0)$ for all Hilbert spaces $\hil$. Finally, denoting by $[n]$ the set $\{1,\ldots,n\}$ whenever $n\in\N$, we define
$$
\mf G(\ms M_0,\ms N_0):=\mf G_{\ell^2_\N}(\ms M_0,\ms N_0)\cup\bigcup_{n=1}^\infty\mf G_{\ell^2_{[n]}}(\ms M_0,\ms N_0).
$$
\end{definition}

We could have defined $\mf G(\ms M_0,\ms N_0)$ as the {\it class} of all POVMs reached through BLMPP using $(\ms M_0,\ms N_0)$ as a resource whatever the target Hilbert space $\hil$ or the target value space $(Z,\mc C)$. However, to make things easier, we now restrict ourselves to target POVMs whose Hilbert space is separable (i.e.\ isomorphic to $\ell^2_{[n]}$ for some $n\in\N$ or to $\ell^2_\N$) and whose value space is $\big(\R,\mc B(\R)\big)$. Especially, our definition caters for the discrete POVMs in a separable Hilbert space. Moreover, from now on, we assume that all the Hilbert spaces, i.e.,\ the fixed Hilbert spaces $\mc K_1$ and $\mc K_2$ and the target Hilbert spaces $\hil$, are separable and the value spaces $(X,\mc A)$ and $(Y,\mc B)$ are standard Borel, i.e., they are $\sigma$-isomorphic to a Borel-measurable subset of a Polish space equipped with the restriction of the Borel $\sigma$-algebra on that subset. Actually, we may view any standard Borel measurable space as a measurable subset of $\R$. This explains our special notational treatment of the space $\big(\R,\mc B(\R)\big)$ in the definition above; we basically treat all POVMs as real POVMs from now on in this section (be they discrete or continuous). Moreover, in this setting, we may assume that the weak Markov kernels are actually ordinary Markov kernels, i.e.,\ the `for almost all' phrases can be replaced with `for all' phrases in the bullet points defining weak Markov kernels. This means that, for (standard Borel) measurable spaces $(X,\mc A)$ and $(Z,\mc C)$, a Markov kernel $\beta(\cdot|\cdot):\mc C\times X\to[0,1]$ is such that $\beta(C|\cdot):X\to[0,1]$ is $\mc A$-measurable for all $C\in\mc C$ and $\beta(\cdot|x):\mc C\to[0,1]$ is a probability measure for all $x\in X$. For (standard Borel) measurable spaces $(X,\mc A)$, $(Y,\mc B)$, and $(Z,\mc C)$ and Markov kernels $\alpha(\cdot|\cdot):\mc B\times X\to[0,1]$ and $\beta(\cdot|\cdot):\mc C\times Y\to[0,1]$, we may define the third Markov kernel $(\beta*\alpha)(\cdot|\cdot):\mc C\times X\to[0,1]$ through
$$
(\beta*\alpha)(C|x)=\int_Y\beta(C|y)\alpha(dy|x),\qquad C\in\mc C,\quad x\in X.
$$
We next introduce some notational conventions.

\begin{definition}
Suppose that $\hil$ is a (separable) Hilbert space and $(X,\mc A)$ is a (standard Borel) measurable space and that $\ms M:\mc A\to\mc L(\hil)$ is a POVM. For another (separable) Hilbert space $\hil'$ and a POVM $\ms N:\mc A\to\mc L(\hil')$ we denote $\ms N\leq_{\rm prae}\ms M$ if there is a channel $\Phi:\mc T(\hil')\to\mc T(\hil)$ such that $\ms N=\Phi^*\circ\ms M$. We call $\ms N$ as a {\it pre-processing of $\ms M$}. For another (standard Borel) measurable space $(Z,\mc C)$ and a POVM $\ms N':\mc C\to\mc L(\hil)$, we denote $\ms N'\leq_{\rm post}\ms M$ if there is a Markov kernel $\beta(\cdot|\cdot):\mc C\times X\to[0,1]$ such that $\ms N'=\beta[\ms M]$. We call $\ms N'$ as a {\it post-processing of $\ms M$}. Let $\hil'$ and $(Z,\mc C)$ be as above. We say that a POVM $\tilde{\ms N}:\mc C\to\mc L(\hil')$ is a {\it hybrid processing of $\ms M$} if there is a channel $\Phi:\mc T(\hil')\to\mc T(\hil)$ and a Markov kernel $\beta(\cdot|\cdot):\mc C\times X\to[0,1]$ such that $\tilde{\ms N}=\beta[\Phi^*\circ\ms M]$, and we denote $\tilde{\ms N}\leq_{\rm h}\ms M$.
\end{definition}

We have already explained the physical meaning of post-processing. Pre-processing means that, before measuring the POVM $\ms M$, we transform the pre-measurement system originally described by the Hilbert space $\hil'$ with a channel $\Phi:\mc T(\hil')\to\mc T(\hil)$. When we next measure $\ms M$, the measurement statistics are given by $\mc A\ni A\mapsto\tr{\Phi(\rho)\ms M(A)}=\tr{\rho\Phi^*\big(\ms M(A)\big)}$ for all input states $\rho\in\mc S(\hil')$, i.e.,\ we actually measure the preprocessing $\Phi^*\circ\ms M$. Hybrid processing is naturally a combination of pre- and post-processing. A natural extension to hybrid processing would be convex combinations of hybrid processing schemes, but we concentrate, for now, on simple hybrid processings. Using the definitions made thus far, we can prove the following elementary observations.

\begin{proposition}\label{prop:Gset}
Let $\ms M$, $\ms N$, $\ms M'$, and $\ms N'$ be any (real; continuous or discrete) POVMs operating in any (separable) Hilbert spaces. The following implications hold:
$$
\begin{array}{rclrclcrcl}
\ms M'&\leq_{\rm prae}&\ms M,&\ms N'&\leq_{\rm prae}&\ms N&\Rightarrow&\mf G(\ms M',\ms N')&\subseteq&\mf G(\ms M,\ms N),\\
\ms M'&\leq_{\rm post}&\ms M,&\ms N'&\leq_{\rm post}&\ms N&\Rightarrow&\mf G(\ms M',\ms N')&\subseteq&\mf G(\ms M,\ms N),\\
\ms M'&\leq_{\rm h}&\ms M,&\ms N'&\leq_{\rm h}&\ms N&\Rightarrow&\mf G(\ms M',\ms N')&\subseteq&\mf G(\ms M,\ms N).
\end{array}
$$
\end{proposition}

\begin{proof}
We only prove the final implication since the two other implications follow from the third one. This is due to the fact that, if $\ms M'\leq_{\rm prae}\ms M$ or $\ms M'\leq_{\rm post}\ms M$, then $\ms M'\leq_{\rm h}\ms M$ (and similarly for $\ms N$ and $\ms N'$). To see this, let us first assume that $\ms M'\leq_{\rm prae}\ms M$ where $\ms M$ operates in the Hilbert space $\mc K$ and $\ms M'$ in $\mc K'$. Let $\Phi:\mc T(\mc K')\to\mc T(\mc K)$ be a channel such that $\ms M'=\Phi^*\circ\ms M$. Define the Markov kernel $\delta(\cdot|\cdot):\mc B(\R)\times\mc R\to[0,1]$ through $\delta(A|x)=\chi_A(x)$ for all $A\in\mc B(\R)$ and $x\in\R$. We now have, for all $A\in\mc B(\R)$,
$$
\delta[\Phi^*\circ\ms M](A)=\int_\R\delta(A|x)(\Phi^*\circ\ms M)(dx)=\int_A(\Phi^*\circ\ms M)(dx)=(\Phi^*\circ\ms M)(A)=\ms M'(A),
$$
i.e.,\ $\delta[\Phi^*\circ\ms M]=\ms M'$, implying that $\ms M'\leq_{\rm h}\ms M$. Assume next that $\ms M'\leq_{\rm post}\ms M$ where both $\ms M$ and $\ms M'$ both operate in the same Hilbert space $\mc K$. Denote the identity map within $\mc T(\mc K)$ by ${\rm id}$. Since there is a Markov kernel $\beta(\cdot
|\cdot):\mc B(\R)\times\R\to[0,1]$ such that $\beta[\ms M]=\ms M'$, we have $\ms M'=\beta[\ms M]=\beta[{\rm id}^*\circ\ms M]$, implying that $\ms M'\leq_{\rm h}\ms M$. Thus, we may concentrate on the final implication.

Let us now assume the left-hand side of the third implication and that $\ms M$, $\ms N$, $\ms M'$, and $\ms N'$ operate, respectively, in the Hilbert space $\mc K_1$, $\mc K_2$, $\mc K'_1$, and $\mc K'_2$. Let $\Phi_i:\mc T(\mc K'_i)\to\mc T(\mc K_i)$, $i=1,\,2$ be channels and $\beta_i:\mc B(\R)\times\R\to[0,1]$, $i=1,\,2$, be Markov kernels such that $\ms M'=\beta_1[\Phi_1^*\circ\ms M]$ and $\ms N'=\beta_2[\Phi_2^*\circ\ms N]$. Define the map $\beta_{1,2}(\cdot|\cdot):\mc B(\R^2)\times\R^2\to[0,1]$ so that $\beta_{1,2}(\cdot|x,y)$ is the product measure of $\beta_1(\cdot|x)$ and $\beta_2(\cdot|y)$ for all $x,\,y\in\R$. Denote the set of sets $C\in\mc B(\R^2)$ such that $\beta_{1,2}(C|\cdot):\R^2\to[0,1]$ is measurable by $\mc C$. One easily sees that $\mc C$ is a Dynkin class which contains the product sets $A\times B$ where $A,\,B\in\mc B(\R)$ and, according to the Dynkin class theorem, $\mc B(\R^2)\subseteq\mc C$, i.e.,\ $\mc C=\mc B(\R^2)$ and $\beta_{1,2}$ is a Markov kernel. Let $\hil$ be a separable Hilbert space and $\ms G:\mc B(\R)\to\mc L(\hil)$ be a POVM within $\mf G(\ms M',\ms N')$ so that there is a channel $\Phi:\mc T(\hil)\to\mc T(\mc K'_1\otimes\mc K'_2)$ and a Markov kernel $\beta(\cdot|\cdot):\mc B(\R)\times\R^2\to[0,1]$ such that, for all $C\in\mc B(\R)$,
\begin{align*}
\ms G(C)=&\beta[\Phi^*\circ(\ms M'\otimes\ms N')](C)=\int_{\R^2}\beta(C|x,y)\Phi^*\big(\ms M'(dx)\otimes\ms N'(dy)\big)\\
=&\int_{\R^2}\int_{\R^2}\beta(C|x,y)\beta_1(dx|z)\beta_2(dy|w)\Phi^*\Big(\Phi_1^*\big(\ms M(dz)\big)\otimes\Phi_2^*\big(\ms N(dy)\big)\Big)\\
=&\int_{\R^2}(\beta*\beta_{1,2})(C|z,w)\big((\Phi_1\otimes\Phi_2)\circ\Phi\big)^*\big(\ms M(dz)\otimes\ms N(dw)\big)\\
=&(\beta*\beta_{1,2})\big[\big((\Phi_1\otimes\Phi_2)\circ\Phi\big)^*\circ(\ms M\otimes\ms N)\big](C),
\end{align*}
implying that $\ms G\in\mf G(\ms M,\ms N)$.
\end{proof}

Let $(X,\mc A)$ be a measurable space, $\hil$ be a Hilbert space, and $\ms M:\mc A\to\mc L(\hil)$ be a POVM. We say that a triple $(\mc M,\ms P,V)$ consisting of a Hilbert space $\mc M$, a PVM $\ms P:\mc A\to\mc L(\mc M)$, and an isometry $V:\hil\to\mc M$ is a {\it Na\u{\i}mark dilation} for $\ms M$ if $\ms M(A)=V^*\ms P(A)V$ for all $A\in\mc A$. A POVM is said to be of {\it rank 1} if it has a Na\u{\i}mark dilation of the form $(L^2_\mu,\ms P_\mu,V)$ where $\mu:\mc A\to\R$ is some positive measure and $\ms P_\mu:\mc A\to\mc L(L^2_\mu)$ is the canonical PVM in $L^2_\mu$, i.e.,\ $\big(\ms P_\mu(A)f\big)(x)=\chi_A(x)f(x)$ for all $A\in\mc A$, $f\in L^2_\mu$, and $x\in X$. The next result tells us that rank-1 PVMs have the highest resource potential for BLMPP.

\begin{theorem}\label{theor:BLMPPmax}
Let $\mc K_1$ and $\mc K_2$ be separable Hilbert spaces and $\ms M_0:\mc B(\R)\to\mc L(\mc K_1)$ and $\ms N_0:\mc B(\R)\to\mc L(\mc K_2)$ be POVMs. There are separable Hilbert spaces $\mc M$ and $\mc N$ and rank-1 PVMs $\ms P_0:\mc B(\R)\to\mc L(\mc M)$ and $\ms Q_0:\mc B(\R)\to\mc L(\mc N)$ such that $\ms M_0\leq_{\rm h}\ms P_0$ and $\ms N_0\leq_{\rm h}\ms Q_0$ and, consequently, $\mf G(\ms M_0,\ms N_0)\subseteq\mf G(\ms P_0,\ms Q_0)$.
\end{theorem}

\begin{proof}
According to \cite{HyPeYl2007}, there are $\sigma$-finite measures $\mu,\,\nu:\mc B(\R)\to[0,\infty]$ such that $\ms M_0$ and, respectively, $\ms N_0$ have (minimal) Na\u{\i}mark dilations $(\mc M,\ms P,V)$ and, respectively, $(\mc N,\ms Q,W)$ where
$$
\mc M=\int^\oplus_\R\mc M(x)\,d\mu(x),\qquad\mc N=\int^\oplus_\R\mc N(y)\,d\nu(y)
$$
where, in turn, $x\mapsto\mc M(x)$ and $y\mapsto\mc N(y)$ are measurable fields of Hilbert spaces. Let us fix measurable fields of bases
$$
\R\ni x\mapsto\{e_k(x)\}_{k=1}^\infty\subset\mc M(x),\qquad\R\ni y\mapsto\{f_\ell(y)\}_{\ell=1}^\infty\subset\mc N(y)
$$
defining the above fields of Hilbert spaces. This means that, for all $x,\,y\in\R$, denoting the dimension of $\mc M(x)$ by $m(x)\in\N\cup\{\infty\}$ and that of $\mc N(y)$ by $n(y)\in\N\cup\{\infty\}$, $\{e_k(x)\}_{k=1}^{m(x)}$ is an orthonormal basis of $\mc M(x)$, $e_k(x)=0$ whenever $k>m(x)$, $\{f_\ell(y)\}_{\ell=1}^{n(y)}$ is an orthonormal basis of $\mc N(y)$, and $f_\ell(y)=0$ whenever $\ell>n(y)$. Above, on one hand, the set of indices $k=1,\ldots,\,0$ and, on the other hand, the set of indices $k>\infty$ are understood as empty (respectively, in the cases when, e.g.,\ $m(x)=0$ and $m(x)=\infty$). Moreover,
$$
\ms P(A)=\int^\oplus_\R\chi_A(x)\id_{\mc M(x)}\,d\mu(x),\qquad\ms Q(B)=\int^\oplus_\R\chi_B(y)\id_{\mc N(y)}\,d\nu(y)
$$
for all $A,\,B\in\mc B(\R)$. Following the complete refinement concept introduced in \cite{Pellonpaa2014}, let us define the rank-1 PVMs $\ms P_0:\mc B(\R\times\N)\to\mc L(\mc M)$ and $\ms Q_0:\mc B(\R\times\N)\to\mc L(\mc N)$ such that
$$
\ms P_0(A\times\{k\})=\int^\oplus_\R\chi_A(x)|e_k(x)\>\<e_k(x)|\,d\mu(x),\qquad\ms Q_0(B\times\{\ell\})=\int^\oplus_\R\chi_B(y)|f_\ell(y)\>\<f_\ell(y)|\,d\nu(y)
$$
for all $A,\,B\in\mc B(\R)$ and $k,\,\ell\in\N$; note that these PVMs are not actually defined on $\mc B(\R)$ but, as their value spaces are standard Borel, they can be turned into real PVMs with simple relabeling.

Define the channels $\Phi_V:\mc T(\mc K_1)\to\mc T(\mc M)$ and $\Phi_W:\mc T(\mc K_2)\to\mc T(\mc N)$ through $\Phi_V(\rho)=V\rho V^*$ and $\Phi_W(\sigma)=W\sigma W^*$ for all $\rho\in\mc T(\mc K_1)$ and $\sigma\in\mc T(\mc K_2)$. Furthermore, define the Markov kernels $\alpha(\cdot|\cdot):\mc B(\R)\times(\R\times\N)\to[0,1]$ and $\beta(\cdot|\cdot):\mc B(\R)\times(\R\times\N)\to[0,1]$ through $\alpha(C|x,k)=\chi_C(x)$ and $\beta(D|y,\ell)=\chi_D(y)$ for all $C,\,D\in\mc B(\R)$, $x,\,y\in\R$, and $k,\,\ell\in\N$. Using the fact that $\sum_{k=1}^\infty\ms P_0(A\times\{k\})=\ms P(A)$ for all $A\in\mc B(\R)$, we now have, for all $C\in\mc B(\R)$,
\begin{align*}
\alpha[\Phi_V^*\circ\ms P_0](C)=&\sum_{k=1}^\infty\int_\R\alpha(C|x,k)V^*\ms P_0(dx\times\{k\})V=\int_\R\chi_C(x)V^*\sum_{k=1}^\infty\ms P_0(dx\times\{k\})V\\
=&\int_C V^*\ms P(dx)V=\int_C\ms M_0(dx)=\ms M_0(C).
\end{align*}
Similarly, $\ms N_0=\beta[\Phi^*_W\circ\ms Q_0]$, so that $\ms M_0\leq_{\rm h}\ms P_0$ and $\ms N_0\leq_{\rm h}\ms Q_0$. The final claim follows now from Proposition \ref{prop:Gset}.
\end{proof}

\subsection{BLMPP protocol for approximate joint POVMs of discrete PVMs with an emphasis on MUBs}\label{subsec:BLMPPMUB}

In this subsection we study, as an example, the BLMPP resource potential of a pair of rank-1 PVMs in finite dimensions where the target observables are approximate joint POVMs for PVMs. Note that, according to Theorem \ref{theor:BLMPPmax}, this pair has a maximal resource potential, i.e.,\ the class of (joint) POVMs reached from it through BLMPP is maximal.

We study PVMs with values $\{1,\ldots,M\}$ and $\{1,\ldots,N\}$ where $M,\,N\in\N$. We treat these sets as $\Z_M$ and $\Z_N$, respectively; this notational convention will soon prove useful. Our target marginal POVMs are PVMs $\ms P=(P_k)_{k\in\Z_M}$ and $\ms Q=(Q_\ell)_{\ell\in\Z_N}$ in some Hilbert space $\hil$. As a resource in the BLMPP protocol we will use the pair $(\ms Q^M,\ms Q^N)$ where, for $L\in\{M,N\}$, $\ms Q^L=(|e^L_k\>\<e^L_k|)_{k\in\Z_L}$ is the rank-1 PVM in $\C^L$ determined by a fixed orthonormal basis $\{e^L_k\}_{k\in\Z_L}\subset\C^L$. Moreover, we define, for all $k\in\Z_L$ ($L\in\{M,N\}$) the unitary $U^L_k=\sum_{i\in\Z_L}|e^L_{i+k}\>\<e^L_i|\in\mc U(\C^L)$. Motivated by the finite standard measurement interaction unitary of Subsection \ref{subsec:finite}, we define the unitaries
$$
U:=\sum_{k\in\Z_M}P_k\otimes U^M_k\otimes\id_N,\qquad V:=\sum_{\ell\in\Z_N}Q_\ell\otimes\id_M\otimes U^N_\ell
$$
on $\hil\otimes\C^M\otimes\C^N$. For any $\sigma\in\mc S(\C^M\otimes\C^N)$, we now define the channel $\Phi_\sigma:\mc T(\hil)\to\mc T(\C^M\otimes\C^N)$ through
$$
\Phi_\sigma(\rho)={\rm tr}_1\big[VU(\rho\otimes\sigma)U^*V^*\big],\qquad\rho\in\mc T(\hil).
$$
This channel can be seen to correspond to a procedure where the target system is coupled to a bipartite ancillary system in the initial state $\sigma$ followed by two consecutive modified standard measurement interactions. This time, we do not post-process, meaning that the joint POVM generated by BLMPP through $\Phi_\sigma$ and using $(\ms Q^M,\ms Q^N)$ as its resource is $\ms G=(G_{k,\ell})_{(k,\ell)\in\Z_M\times\Z_N}$,
\begin{align*}
G_{i,j}=&\Phi_\sigma^*(|e^M_i\>\<e^M_i|\otimes|e^N_j\>\<e^N_j|)\\
=&\sum_{k,k'\in\Z_M}\sum_{\ell\in\Z_N}\tr{\sigma(|e^M_{i-k}\>\<e^M_{i-k'}|\otimes|e^N_{j-\ell}\>\<e^N_{j-\ell}|)}P_k Q_\ell P_{k'}
\end{align*}
for all $i\in\Z_M$ and $j\in\Z_N$ where the final form is obtained through straightforward calculation.

Denoting the first margin of $\sigma$ as $\sigma_1$, one easily sees that the first margin $\ms G^{(1)}=(G^{(1)}_i)_{i\in\Z_M}$ of the joint POVM $\ms G$ is given by
$$
G^{(1)}_i=\sum_{j\in\Z_N}G_{i,j}=\sum_{k\in\Z_M}\<e^M_{i-k}|\sigma_1 e^M_{i-k}\>P_k
$$
for all $i\in\Z_M$, i.e.,\ $\ms G^{(1)}$ is the convolution $p^{\ms Q^M}_{\sigma_1}*\ms P$ of $\ms P$ by the $\ms Q^M$-distribution in $\sigma_1$. The second margin $\ms G^{(2)}=(G^{(2)}_j)_{j\in\Z_N}$ is not that simple: To write down the general form of this margin, let us fix a spectral decomposition $\sigma_1=\sum_{m\in\Z_M}s_m|\eta_m\>\<\eta_m|$ and define the operators
$$
K_{m,n}:=\sqrt{s_m}\sum_{k\in\Z_M}\<e^M_{n-k}|\eta_m\>P_k,\qquad m,\,n\in\Z_M.
$$
Using these definitions and notations, one easily finds that
$$
G^{(2)}_j=\sum_{i\in\Z_M}G_{i,j}=\sum_{m,n\in\Z_M}K_{m,n}^*(p^{\ms Q^N}_{\sigma_2}*\ms Q)_j K_{m,n}
$$
for all $j\in\Z_N$ where $(p^{\ms Q^N}_{\sigma_2}*\ms Q)_j=\sum_{\ell\in\Z_N}\<e^N_{j-\ell}|\sigma_2 e^N_{j-\ell}\>Q_\ell$ and $\sigma_2$ is the second margin of $\sigma$.

Let us now assume that $\sigma_1=|h^M\>\<h^M|$ where $h^M=\sqrt{\sqrt{M}/\big(2(\sqrt{M}+1)\big)}(e^M_0+f^M_0)$ where, in turn, $f^M_0=M^{-1/2}(e^M_0+\cdots+e^M_{M-1})$. It now follows that $K_{m,n}=\delta_{m,0}K_n$ where $K_n=\sqrt{\sqrt{M}/\big(2(\sqrt{M}+1)\big)}(P_n+M^{-1/2}\id_M)$. Using this, one easily arrives at
\begin{align}
\ms G^{(1)}=&\frac{\sqrt{M}+2}{2(\sqrt{M}+1)}\ms P+\frac{\sqrt{M}}{2(\sqrt{M}+1)}\ms T\label{eq:uusiG1}\\
\ms G^{(2)}=&\frac{\sqrt{M}+2}{2(\sqrt{M}+1)}p^{\ms Q^N}_{\sigma_2}*\ms Q+\frac{\sqrt{M}}{2(\sqrt{M}+1)}\ms Q^{\ms P},\label{eq:uusiG2}
\end{align}
where $\ms T=(T_i)_{i\in\Z_M}$ is the trivial observable determined by $T_i=(1/M)\id_\hil$ for all $i\in\Z_M$ and $\ms Q^{\ms P}=(Q^{\ms P}_j)_{j\in\Z_N}$ is defined through $Q^{\ms P}_j=\sum_{i\in\Z_M}P_i Q_j P_i$ for all $j\in\Z_N$. If $\sigma_2=|e^N_0\>\<e^N_0|$, then $p^{\ms Q^N}_{\sigma_2}$ is supported by $\{0\}$ and Equations \eqref{eq:uusiG1} and \eqref{eq:uusiG2} imply that $\ms G^{(1)}$ is a mixture of $\ms P$ and $\ms T$, where the target PVM $\ms P$ appears with the weight $w_M:=(\sqrt{M}+2)/\big(2(\sqrt{M}+1)\big)$, and $\ms G^{(2)}$ is a mixture of $\ms Q$ and $\ms Q^{\ms P}$ where $\ms Q$ appears also with the weight $w_M$.

Let us lastly consider the case where $\dim{\hil}=:D<\infty$ and $\ms P$ and $\ms Q$ are maximally incompatible in the sense that they are associated to mutually unbiased bases (MUBs). This means that there are orthonormal bases $\{\fii_i\}_{i\in\Z_D},\,\{\psi_j\}_{j\in\Z_D}$ such that $|\<\fii_i|\psi_j\>|=D^{-1/2}$ for all $i,\,j\in\Z_D$ such that $i\neq j$ and $\ms P=(|\fii_i\>\<\fii_i|)_{i\in\Z_D}$ and $\ms Q=(|\psi_j\>\<\psi_j|)_{j\in\Z_D}$. Now $M=D=N$ and we still assume that $\sigma_1=|h^D\>\<h^D|$ and $\sigma_2=|e^D_0\>\<e^D_0|$ (whence $\sigma=\sigma_1\otimes\sigma_2$). Using the MUB condition, $\ms Q^{\ms P}=\ms T$ is the constant trivial observable, so that Equations \eqref{eq:uusiG1} and \eqref{eq:uusiG2} imply
\begin{align*}
\ms G^{(1)}=&\frac{\sqrt{D}+2}{2(\sqrt{D}+1)}\ms P+\frac{\sqrt{D}}{2(\sqrt{D}+1)}\ms T=\frac{1}{2}\Big(1+\frac{1}{\sqrt{D}}\Big)\ms P+\frac{1}{2}\Big(1-\frac{1}{\sqrt{D}}\Big)\ms M\\
\ms G^{(2)}=&\frac{\sqrt{D}+2}{2(\sqrt{D}+1)}\ms Q+\frac{\sqrt{D}}{2(\sqrt{D}+1)}\ms T=\frac{1}{2}\Big(1+\frac{1}{\sqrt{D}}\Big)\ms Q+\frac{1}{2}\Big(1-\frac{1}{\sqrt{D}}\Big)\ms N
\end{align*}
where $\ms M=(M_i)_{i\in\Z_D}$ is defined through $M_i=(D-1)^{-1}(\id_D-|\fii_i\>\<\fii_i|)$ for all $i\in\Z_D$ and $\ms N=(N_j)_{j\in\Z_D}$ is defined through $N_j=(D-1)^{-1}(\id_D-|\psi_j\>\<\psi_j|)$ for all $j\in\Z_D$. Thus, $\ms G$ provides the optimal approximate joint measurement for $\ms P$ and $\ms Q$ since its marginals correspond to mixing the least amount of noise to the targets $\ms P$ and $\ms Q$ so that the mixtures become jointly measurable; see the proof of Theorem 4 of \cite{Haapasalo2015} and Section III.G of \cite{DeFaKa2019}.

\section{Conclusions}

We have investigated using broadcasting channels for generating (approximate) joint measurements for POVMs. For two PVMs $\ms P$ and $\ms Q$, this strategy is successful in essentially all cases when the value spaces are discrete. However in the continuous case, we need to assume that the target POVM is absolutely continuous w.r.t.\ the product PVM $\ms P\otimes\ms Q$. We also derive a no-go result: the canonical joint measurement of a continuous PVM with itself cannot be generated through broadcasting. In the case of fuzzy POVMs, the situation is more difficult, but we are able to derive necessary and sufficient conditions for broadcastability when the POVMs to be jointly measured are mixtures of PVMs with coloured trivial noise. We investigate the realization of covariant phase space observables through broadcasting where the local measurements are the canonical position and momentum observables. We have also investigated what restrictions symmetry requirements place on broadcastability. Finally, we have introduced and studied scenarios where we are allowed to broadcast the quantum system, then do local measurements, and finally classically post-processing the outcome data. In this setting, we have identified rank-1 PVMs as the most resourceful local measurements.

Despite the no-go result for canonical joint measurements of continuous PVMs, it remains a possibility that such continuous joint measurements could be generated approximately through broadcasting with a vanishing error. However, the suitable approximations and their quantification remain to be identified. Moreover, our most explicit and easiest to check conditions for broadcastability of local fuzzy POVMs are only valid for mixtures of PVMs and classical noise; quite a limited set of local measurements. Finding similar simple conditions for more general fuzzy POVMs remains an interesting problem. In our current formulation of the BLMPP protocol, we allow any broadcasting channels and local measurements. However, broadcasting is a rather expensive resource. In realistic laboratory settings, we might have access to only restricted sets of broadcasting channels and local measurements. It is thus crucial to find the most fruitful local measurements for the BLMPP protocol in realistic restricted settings. It seems likely that, in general, rank-I PVMs might no longer be the best local measurements one could use, but this requires further study.

\section*{Acknowledgements}

\noindent{}This work is funded by the National Research Foundation, Prime Minister's Office, Singapore and the Ministry of Education, Singapore under the Research Centres of Excellence programme.

\end{document}